\newif\ifAnon
\newif\ifArxiv
\Arxivtrue

\ifArxiv
\documentclass{article}
\usepackage[backref=page]{hyperref}
\hypersetup{
    colorlinks=true,
    linkcolor=violet,     
    urlcolor=cyan,
    pdftitle={Degree Sequence Reconstruction from Subgraph Traces},
    }

\usepackage[a4paper, total={6in, 8in}]{geometry}
\else

\documentclass[11pt,nonacm, review=true, backref=page]{acmart}
\authorsaddresses{}
\fi
\usepackage{aricksMacros}

\usepackage{microtype}

\usepackage{amsfonts}
\usepackage{graphicx}
\usepackage{algorithm}
\usepackage{algorithmic}
\usepackage[most]{tcolorbox}

\usepackage{cleveref}

\newtcbtheorem[number within=section]{theorembox}{Theorem}{
    fonttitle=\bfseries,   % Title font styling
    breakable,             % Allows splitting across pages
    enhanced,              % Enables advanced styling configurations
}{thm}

\newtcbtheorem[number within=section]{propositionbox}{Proposition}{
    fonttitle=\bfseries,   % Title font styling
    breakable,             % Allows splitting across pages
    enhanced,              % Enables advanced styling configurations
}{prop}

\newtcbtheorem[number within=section]{lemmabox}{Lemma}{
    fonttitle=\bfseries,   % Title font styling
    breakable,             % Allows splitting across pages
    enhanced,              % Enables advanced styling configurations
}{lem}

\newtcbtheorem[number within=section]{corollarybox}{Corollary}{
    fonttitle=\bfseries,   % Title font styling
    breakable,             % Allows splitting across pages
    enhanced,              % Enables advanced styling configurations
}{cor}

\newtcbtheorem[number within=section]{remarkbox}{Remark}{
    fonttitle=\bfseries,   % Title font styling
    breakable,             % Allows splitting across pages
    enhanced,              % Enables advanced styling configurations
}{rem}

\makeatletter
\crefname{tcb@cnt@lemmabox}{lemma}{lemmas}
\Crefname{tcb@cnt@lemmabox}{Lemma}{Lemmas}
\crefformat{tcb@cnt@lemmabox}{lemma~#2#1#3}
\Crefformat{tcb@cnt@lemmabox}{Lemma~#2#1#3}

\crefname{tcb@cnt@theorembox}{theorem}{theorems}
\Crefname{tcb@cnt@theorembox}{Theorem}{Theorems}
\crefformat{tcb@cnt@theorembox}{theorem~#2#1#3}
\Crefformat{tcb@cnt@theorembox}{Theorem~#2#1#3}

\crefname{tcb@cnt@corollarybox}{corollary}{corollaries}
\Crefname{tcb@cnt@corollarybox}{Corollary}{Corollaries}
\crefformat{tcb@cnt@corollarybox}{corollary~#2#1#3}
\Crefformat{tcb@cnt@corollarybox}{Corollary~#2#1#3}

\crefname{tcb@cnt@propositionbox}{proposition}{propositions}
\Crefname{tcb@cnt@propositionbox}{Proposition}{Propositions}
\crefformat{tcb@cnt@propositionbox}{proposition~#2#1#3}
\Crefformat{tcb@cnt@propositionbox}{Proposition~#2#1#3}

\crefname{tcb@cnt@remarkbox}{remark}{remarks}
\Crefname{tcb@cnt@remarkbox}{Remark}{Remarks}
\crefformat{tcb@cnt@remarkbox}{remark~#2#1#3}
\Crefformat{tcb@cnt@remarkbox}{Remark~#2#1#3}

\makeatother

\title{Degree Sequence Reconstruction from Subgraph Traces}

\ifAnon
    
\else

\ifArxiv
\author{
Venkata Gandikota%
\thanks{Syracuse University, vsgandik@syr.edu}%
\and
Arick Grootveld%
\thanks{Syracuse University, aegrootv@syr.edu}%
\footnotemark[2]
\and
Haodong Yang%
\thanks{Syracuse University, hyang85@syr.edu}%
\footnotemark[2]
\thanks{All authors contributed equally to this work and share first authorship. Authors are listed in alphabetical order by surname.}
}

\date{}

\else
\author{Venkata Gandikota}
\affiliation{\institution{Syracuse University} \country{United States}}
\email{vsgandik@syr.edu}

\author{Arick Grootveld}
\affiliation{\institution{Syracuse University} \country{United States}}
\email{aegrootv@syr.edu}

\author{Haodong Yang}
\affiliation{\institution{Syracuse University} \country{United States}}
\email{hyang85@syr.edu}

\fi

\fi

\DeclareMathOperator{\diag}{diag}

\ifAnon
    \hypersetup{
        pdftitle={\small{Degree Sequence Reconstruction from Random Subgraphs}},
        pdfauthor={},
    }
\else
    \hypersetup{
        pdftitle={Degree Sequence Reconstruction from Random Subgraphs},
        pdfauthor={Venkata Gandikota, Arick Grootveld, Haodong Yang},
    }
\fi

\makeatletter
\let\@authorsaddresses\@empty
\makeatother

\begin{document}

\ifArxiv

\else
\begin{CCSXML}
    <ccs2012>
    <concept>
    <concept_id>10002950.10003648</concept_id>
    <concept_desc>Mathematics of computing~Probability and statistics</concept_desc>
    <concept_significance>300</concept_significance>
    </concept>
    <concept>
    <concept_id>10002950.10003712</concept_id>
    <concept_desc>Mathematics of computing~Information theory</concept_desc>
    <concept_significance>300</concept_significance>
    </concept>
    </ccs2012>
\end{CCSXML}

\ccsdesc[300]{Mathematics of computing~Probability and statistics}
\ccsdesc[300]{Mathematics of computing~Information theory}

\keywords{Graph Reconstruction, Trace Reconstruction, Sample Complexity}

\fi

\ifArxiv
\maketitle
\fi

\begin{abstract}
    The goal of degree sequence reconstruction is to recover the ordered vector of degrees of an unknown graph from vertex deleted traces, where each vertex is deleted independently with probability $p$. We provide two algorithms for reconstruction; the first uses rejection sampling to reduce the problem to an estimation problem for a mixture distribution. Combined with prior trace reconstruction results, this gives a reconstruction algorithm using $\Exp{\tilde O (n^{1/3})}$ traces, although no sub-exponential time decoder is known. Our other approach involves recovering certain graph invariants,  degree moments, that can identify a graphs degree sequence. Extremal polynomial bounds show that $\tilde \Theta(n^{1/2})$ degree moments are necessary and sufficient to reconstruct the degree sequence, which leads to an algorithm with $\Exp{\tilde O(n^{1/2})}$ trace complexity. The same polynomial machinery yields a sub-exponential time decoder for the degree sequence from the moments. Additionally, we give an $O(n^{3})$ upper bound and a $\Omega(n^2)$ lower bound for the trace complexity of recovering the number of edges. 
\end{abstract}

\ifArxiv
\else 
\maketitle
\fi

% \newpage

\ifArxiv
\tableofcontents
\fi

\section{Introduction}
\label{sec:introduction}

Graph reconstruction is a problem with a long history of interest in graph theory and combinatorics, due to Kelly  \cite{kelly1942isometric} and Ulam \cite{Ulam1960ACollection}. In the simplest version of the problem, we have a graph $G$, and we are permitted to look at the collection of $n$ cards, which are the subgraphs made by removing a single vertex in $G$. 
The goal is to reconstruct $G$ from its cards, and the famous reconstruction conjecture claims that all graphs are reconstructible from their deck. A recent result by Ivanov \cite{ivanov2026nonisomorphic} reports that there are families of graphs which are indistinguishable using $\alpha n$ cards, for any $\alpha \in (0,1)$, meaning that no fixed fraction of cards is sufficient to reconstruct an arbitrary graph. In contrast, Bollobás \cite{Bollobas1990AlmostEvery} showed that almost all  graphs can be reconstructed using only $3$ cards. For more information on graph reconstruction, see one of the survey papers by Harary \cite{Harary1974ASurvey}, Bondy and Hemminger \cite{Bondy1977GraphRecon}, or Asciak, Francalanza, Lauri, and Myrvold \cite{asciak2010survey}.

Another reconstruction problem of theoretical interest is \textit{trace reconstruction}, first studied by Kalashnik \cite{kalashnik1973reconstruction} and Levenshtein \cite{Lev97}. The goal in the trace reconstruction problem is to recover a length $n$ binary string from random sub-strings (traces) which are generated by deleting symbols independently at random with probability $p := 1-q$. Mean based algorithms are known to have trace complexity (and runtime) $\Exp{\tilde \Theta(n^{1/3})}$ due to De, O'Donnell and Servedio \cite{de2017optimal} as well as Nazarov and Peres \cite{Nazarov2017TraceRecon}. Cheng, Grigorescu, Li, Sudan and Zhu \cite{Cheng2025OnKMer} showed that $k$-mer based algorithms have trace complexity $\Exp{\tilde \Theta(n^{1/5})}$, and the maximum likelihood estimator has trace complexity within a factor of $n$ of any optimal method. For arbitrary algorithms a recent result by Burudgunte, Valiant and Wang \cite{burudgunte2026quasipolynomial} gives the best known upper bound of $\Exp{ O(\log^c(n)) }$ for a constant $c > 0$, while the best lower bound is due to Chase \cite{chase2020newlowerboundstrace} $\tilde \Omega\left(n^{3/2}\right)$. 
There is a large body of work on variations of trace reconstruction, including coded trace reconstruction \cite{Cheraghchi2020Coded,Brakensiek2020CodedConstant,Srinivasavaradhan2021Trellis, Hanna2022CodingFor,Kas2023OptimalCodes,Rathore2025ReedMuller},  approximate trace reconstruction \cite{chase2021approximate, Xi2023Approximate}, and circular trace reconstruction \cite{Narayanan2021Circular,burudgunte2025Circular}. 
For randomly generated binary strings Holden, Pemantle, and Peres \cite{Holden2018Subpolynomial} proved an upper bound of $\Exp{O(\log^{1/3}n)}$, which was subsequently improved to $\Exp{\tilde O(\log^{1/5} n)}$ by Rubinstein \cite{Rubinstein2023AverageCase}. 

McGregor and Sengupta \cite{mcgregor2022graph,McGregor2024GraphRecon} proposed a bridge between these two problems, called the \textit{graph trace reconstruction problem}. 
In this problem, a random subgraph trace is obtained by deleting vertices independently at random with probability $p$. The goal then is to reconstruct an  unknown graph $G$ from the fewest number of sampled traces. 
McGregor and Sengupta \cite{mcgregor2022graph} showed that the number of traces required to reconstruct a randomly generated graph is $\Theta(q^{-2}\log n )$, and in \cite{McGregor2024GraphRecon} they generalized this result to a setting where edges could be added or removed independently at random. Additionally, they found that for arbitrary graphs $\Exp{\Omega(n)}$ traces are necessary to reconstruct, which coincides with the trivial upper bound of $\Exp{O(n)}$ from waiting for a single trace with no vertex deletions. This matches the story of the graph-and-trace reconstruction problem: randomly generated objects can be reconstructed with a small amount of information, while reconstructing arbitrary objects is significantly more challenging. 
Notable variations of graph trace reconstruction include tree trace reconstruction which was introduced by Davies, Rácz, and Rashtchian \cite{davies2019reconstructing} and further developed by Maranzatto \cite{Maranzatto2020TreeTrace, maranzatto2022reconstructing, Maranzatto2024TreeTrace}, and Brailovskaya and Rácz \cite{Brailovskaya_Racz_2023}; as well as spider graph reconstruction by Sun and Yue \cite{Sun2023TraceReconSpiderGraphs}.

Between the two extremes are questions about the number of cards required to recover a graph's parameters. For example, Brown and Fenner \cite{Brown2018TheSize} showed that $n-2$ cards are sufficient to recover the number of edges in $G$, which was later improved by Groenland, Guggiari, and Scott \cite{Groenland2021Size} to $n - \frac{1}{20} \sqrt{n}$. Myrvold \cite{myrvold1992degree} showed that $n-1$ cards are sufficient to reconstruct the degree sequence of a graph. Groenland, Johnston, Kupavskii, Meeks, Scott and Tan showed that for a graph with average degree $d$, $n - O(n/d^3)$ cards suffice to reconstruct the degree sequence \cite{Groenland2022Reconstructing}. Bowler, Brown, Fenner and Myrvold \cite{bowler2011recognizing} showed that $\left\lfloor \frac{n}{2}\right \rfloor + 2$ cards can be used to determine the connectedness of a graph. 
One version of the reconstruction problem considers what families or properties of graphs can be reconstructed from the $\ell$-deck, that is the multiset of all $\ell$-vertex induced subgraphs. Taylor \cite{Taylor1990ReconstructingDegSeqs} showed that for $\ell \sim n(1- \frac1e)$, the $\ell$-deck is sufficient to reconstruct the degree sequence of a graph. Groenland, Johnston, Scott and Tan \cite{Groenland2026SmallerCards} further refined this, showing that the $\ell$-deck can reconstruct the degree sequence when $\ell \geq \sqrt{2 n \log (2n)}$, along with results that the connectedness of a graph can be determined if $\ell \geq \frac{9n}{10}$, and that arbitrary trees can be reconstructed if $\ell \geq \frac{8n}{9} + \frac{4}{9}\sqrt{8n + 5} + 2$. 
In the graph trace setting, however, the example of McGregor and Sengupta \cite{mcgregor2022graph} demonstrates that the connectedness of a graph requires $\Exp{\Omega(n)}$ traces to determine. 

In this work, we study the problem of reconstructing the degree sequence of a graph from traces. This problem has been considered in the statistical literature, first by  \cite{Frank1980Estimation}, later \cite{zhang2015estimating} studied it from the perspective of social network monitoring, and by
\cite{garrard2017goodness} for protein interaction networks. These works focused on the statistical properties of estimators, and gave results for approximation and testing problems. 
In contrast, we address the sample complexity of exact degree sequence reconstruction. The moment matching bound in Section \ref{subsec:EdgeMoment_NumberOfRequiredMoments} is closely related to the result of Scott \cite{scott1997reconstructing} on reconstructing a sequence from the multiset of subsequences, as well as the Prouhet–Tarry–Escott problem \cite{Borwein2002ComputationalExcursions}. In particular, \cite[Lemma 11]{Groenland2026SmallerCards} yields a sharper upper bound than that provided in \Cref{thm:upperBound_kn}. See \Cref{rem:RelationToScottsWork} for more details. 

% Our upper bound for the number of moments required to reconstruct our degree sequence are similar in nature to the work of Scott \cite{scott1997reconstructing} who explored multiset reconstruction from subsequences, and Groeland, Johnston, Scott and Tan \cite{Groenland2026SmallerCards} who considered reconstructing trees from their $(n-r)$ sized cards. 

\subsection{Our Results}
As a warm-up, we show that $\tilde O (n^3)$ traces suffice to recover the number of edges (Theorem~\ref{thm:EdgeCountRec_Achievability}), and that $\tilde \Omega(n^2)$ traces are necessary (Theorem~\ref{thm:LowBound_EdgeCount}). 
For degree sequence reconstruction, we provide two methods. The first (Theorem~\ref{thm:binomMixtureTraceParamRecon}) recovers the degree sequence with $m = \Exp{\tilde O(n^{1/3})}$ traces by using rejection sampling to convert the problem to estimating the parameters of a mixture of binomials, and then using the result of Krishnamurthy, Mazumdar, McGregor and Pal\cite{krishnamurthy2019trace, Krishnamurthy2020Algebraic} from trace reconstruction. While the result gives an information theoretic separation between graphs with different degree sequences, there is no known sub-exponential runtime algorithm for reconstruction. 
The second method (Theorem~\ref{thm:RecDegSeq_FromTraces}) uses a moment based approach to recover the degree sequence from $m = \Exp{\tilde O(\sqrt{n})}$ traces, and we provide an algorithm with matching runtime. The moment based approach involves estimating the sample complexity of recovering certain invariants of the graph, and then demonstrating how many of these invariants are sufficient to identify the degree sequence of the graph uniquely. 
% In Theorem~\ref{}, we also show the optimality (up to $\log$ factors) of the trace complexity in Theorem~\ref{thm:RecDegSeq_FromTraces} for the degree-moment based methods.

% Table~\ref{tab:resultsTable} compares the existing rejection-sampling approach with our degree-moment method, showing that our approach achieves a sub-exponential runtime of \(\exp\{\tilde O(n^{1/2})\}\) while using \(\exp\{\tilde O(n^{1/2})\}\) samples.

%An outline of our contributions can be found in \Cref{tab:resultsTable}. 

\paragraph{Graph Size Reconstruction} For the upper bound on the trace complexity, we show that the maximum variance of the number of edges in the graph trace is bounded by $O(n^3)$, which is achieved by the complete graph. Applying Bernstein's inequality gives us a sample complexity of $O\left(\frac{n^3}{q^2}\log(1/\delta)\right)$. 
For the lower bound, we consider complete bipartite graphs on $\left(\frac{n}{2}, \frac{n}{2}\right)$ and $\left(\frac{n}{2}-1, \frac{n}{2}+1\right)$ vertices. We show that the squared Hellinger distance after applying the vertex deletion channel is $O\left(n^{-2}\right)$, which leads to a sample complexity of $\Omega(n^2)$ to distinguish between the two graphs with high probability.

% \noindent \textit{Degree Sequence Reconstruction}: Suppose you sample a vertex uniformly at random from each trace and take its degree. Conditioned on the vertex surviving the channel, the resulting degree is a binomial random variable, so that the the sampled degree follows a binomial mixture distribution. With this in mind, our first method uses rejection sampling to convert traces to binomial mixture samples, with mixing weights that describe the degree sequence. Applying \cite[Theorem 8]{krishnamurthy2019trace} shows that the mixture parameters can be estimated using $\Exp{\tilde O(n^{1/3})}$ samples, leading to \Cref{thm:binomMixtureTraceParamRecon}. To the best of our knowledge, the only method for recovering the mixture parameters is the trivial algorithm of searching over all mixture parameters, which has an $\Exp{O(n)}$ runtime.

% Our other method involves reconstructing the degree sequence using graph invariants that we call `degree moments', which are linear functionals 

\paragraph{Degree Sequence Reconstruction} Using rejection sampling, we reduce the problem of reconstructing the degree sequence of a graph to a parameter estimation problem for mixtures of binomial distributions. Applying \cite[Theorem 8]{krishnamurthy2019trace} shows that this parameter estimation problem can be solved using $\Exp{\tilde O(n^{1/3})}$ samples, leading to \Cref{thm:binomMixtureTraceParamRecon}. To the best of our knowledge, the only method for recovering the mixture parameters is the trivial algorithm of searching over all mixture parameters, which has an $\Exp{O(n)}$ runtime. 

We provide another algorithm to reconstruct the degree sequence using a graph invariant that we call `degree moments'. In \Cref{thm:TraceComplexity_EdgeMomentEstimation} we show that $\Exp{\tilde O (k)}$ suffice to estimate the first $k$ degree moments of $G$ with high probability. Using the polynomial bound from \cite[Theorem 2.2]{erdelyi2016coppersmith}, we prove \Cref{thm:upperBound_kn}, which shows that $k = O\left(\sqrt{n \log n}\right)$ degree moments are sufficient to recover the degree sequence of any graph. Using a pigeonhole argument, along with the Erdös-Gallai theorem \cite{erdos1960grafok,Choudum_1986}, we construct graphs with different degree sequences and matching first $k = \Omega\left(\sqrt{\frac{n}{\log n}}\right)$ degree moments, which gives us \Cref{thm:lowerBound_kn}. In combination, this shows that $k = \tilde \Theta(\sqrt{n})$ degree moments are necessary and sufficient to recover the degree sequence. Finally, we develop an algorithm to peel off elements of the degree sequence using degree moments by constructing a suitable polynomial, leading to \Cref{thm:RecDegSeq_FromMoments}. 
% \begin{table}[h]
% \begin{tabular}{|c|c|c|}
% \hline
%  & \textbf{Sample Complexity} & \textbf{Runtime} \\ \hline
% % \textit{\begin{tabular}[c]{@{}c@{}}Graph Size:\\ Upper Bound\end{tabular}}  & $O(n^3)$                   & $\poly(n)$         \\ \hline
% % \textit{\begin{tabular}[c]{@{}c@{}}Graph Size: \\ Lower Bound\end{tabular}} & $\Omega(n^2)$              & ---              \\ \hline
% \textit{\begin{tabular}[c]{@{}c@{}}Degree Sequence:\\ Rejection Sampling\end{tabular}} &
%   $\exp\left\{ \tilde O\left(n^{1/3}\right) \right\}$ &
%   $\exp\{O(n)\}$ \\ \hline
% \textit{\begin{tabular}[c]{@{}c@{}}Degree Sequence:\\ Degree Moments\end{tabular}} &
%   $\exp\left\{ \tilde O\left(n^{1/2}\right) \right\}$ &
%   $\exp\left\{ \tilde O\left(n^{1/2}\right) \right\}$ \\ \hline
% \end{tabular}
% \caption{Table of sample complexities and runtime results for degree sequence reconstruction (fixed $q$)}
% \label{tab:resultsTable}
% \end{table}
~\\

The rest of the paper is organized in the following manner: \Cref{sec:problem-setup} describes the vertex deletion channel, the degree sequence reconstruction problem, and the degree moments of a graph. 
\Cref{sec:EdgeCountRecon} gives an upper and lower bound on the number of traces required to recover the number of edges in a graph. \Cref{sec:DegSeqRecon_BinomMixtures} describes a method for reconstructing the degree sequence from the binomial mixture result \cite[Theorem 8]{krishnamurthy2019trace}. \Cref{sec:DegSeqRecon_EdgeMoments} gives an algorithm for reconstructing the degree sequence using degree moments. In \Cref{sec:conclusion} we discuss the results and what problems remain.  

%%%%%%%%%%%%%%%%%%%%%%%%%%%%%%%%%%%%%%%%%%%%%%%%
\section{Preliminaries}
\label{sec:problem-setup}
All graphs in this work are assumed to be unlabeled, finite, and simple.  
We consider the number of vertices, $n$, to be known\footnote{For $n$ unknown, we can recover the parameter exactly with high probability using standard techniques in statistics; For example, see the recent work by Georgieva and Vidakovic \cite{Georgieva2025Revisiting}, and references therein. To recover $n$ with probability $> 1- \delta$, $m = O\left(n \log(1/\delta)\right)$ traces would suffice. Since all problems addressed in this work require $\Omega(n^2)$ traces, we could assume $n$ to be unknown without loss of generality.}, and the vertex retention probability, $q$ to be fixed and known. For $a \in \R$, $\lceil a \rfloor$ denotes rounding to the nearest integer. We use $\log$ to denote the natural logarithm, while $\log_2$ denotes the logarithm base 2.
Let $\R_k[x]$ denote the set of polynomials with real coefficients and degree $\leq k$, and $\mb C_k[x]$ the set of polynomials with complex coefficients and degree $\leq k$. 

Take $x^{\underline{a}}$ to be the falling factorial power of $x$, so that 
\begin{equation*}
    x^{\underline{a}} := x(x-1) \cdots (x-a+1).
\end{equation*}
The first $k$ falling factorial moments of $x$ form a basis for $\R_k[x]$, and the Stirling numbers of the first kind, $s(a,r)$, are the coefficients to convert between the standard and falling factorial polynomial basis. By this we mean, 
\begin{equation}
    \label{eq:StirlingNumbersFirstKind}
    x^{\underline{a}} = \sum_{r=0}^{a} s_{a,r} x^r.
\end{equation}
Similarly, the Stirling numbers of the second kind, $\bfrac{a}{r}$, are the coefficients to convert from the falling factorial to the standard basis, so that 
\begin{equation}
    \label{eq:StirlingNumbersSecondKind}
    x^a = \sum_{r=0}^a \bfrac{a}{r} x^{\underline{r}}.
\end{equation}
We use the convention $\bfrac{0}{0} = 1$ and $\bfrac{r}{j} = 0$ whenever $j < 0$, $j > r$ or $j=0 < r$.

Let $\mc G_n$ denote the set of unlabeled simple graphs on $n$
vertices, and let $\mc G_n^{\leq}$ denote the set of unlabeled simple
graphs on at most $n$ vertices. For a graph $G$, let $V(G)$ and
$E(G)$ denote its vertex and edge sets, respectively. Fix a vertex-deletion probability $p\in(0,1)$, and let
\[
    q:=1-p
\]
denote the vertex-retention probability. For a graph
$G\in\mc G_n$, let
\[
    \Phi_p(G)\in\mc G_n^{\leq}
\]
be the random induced subgraph obtained by retaining every vertex
independently with probability $q$.
The graph $\Phi_p(G)$ is the induced subgraph of $G$ on the retained vertex set 
\[
    \left\{v \in V(G): Z_v = 1\right\},
\]
with $Z_v \sim \rm{Bernoulli}(q)$ the random variable describing whether $v$ is retained after applying the channel.   

For any graph $H\in\mc G_n^{\leq}$, define its degree-count vector by
\[
    d_H
    :=
    \left(
        d_{H,0},d_{H,1},\ldots,d_{H,n-1}
    \right)
    \in\mb Z_{\geq0}^{n},
\]
where $d_{H,i}$ is the number of vertices of degree $i$ in $H$.
Thus,
\begin{equation}
    \sum_{k=0}^{n-1}d_{H,k} = |V(H)|
    \quad\text{and}\quad
    \sum_{k=0}^{n-1} k\cdot d_{H, k} = 2\abs{E(H)},
\end{equation}
We refer to $d_H$ as the degree sequence of $H$, although it is more
precisely the degree-count representation of the usual ordered degree
sequence.

\noindent Let
\begin{equation*}
    \mc D_n := \left\{ d_G: G\in\mc G_n \right\}
\end{equation*}
denote the set of graphical degree-count vectors on exactly $n$
vertices. Similarly, define
\begin{equation*}
    \mc D_n^{\leq} := \left\{d_H: H\in\mc G_n^{\leq}\right\}.
\end{equation*}
We use the relaxed set
\begin{equation}
    \label{eq:relaxedDegSeqSet}
    \mc H_n := \left\{d\in\mb Z_{\geq0}^{n}:\sum_{i=0}^{n-1}d_i=n\right\}.
\end{equation}
Every graphical degree-count vector belongs to $\mc H_n$, so
\begin{equation*}
    \mc D_n\subseteq\mc H_n,
\end{equation*}
and for $n \geq 3$ this inclusion is strict. 

\subsection{The Degree-Sequence Deletion Channel}
\label{subsec:degree-sequence-channel}

Although vertex deletion naturally produces the random induced subgraph $\Phi_p(G)$, the parameters we estimate are invariant to the graph's topology. 
The algorithms proposed in this work only utilize the degree-count vector of each subgraph trace. Therefore, we consider a slightly weaker observation model that can be derived from subgraph traces. We will assume access to \textit{degree-sequence traces} that are defined as the degree sequence of a subgraph trace. 

Define the degree-sequence deletion channel by
\begin{equation*}
    \phi_p(G) := d_{\Phi_p(G)}.
\end{equation*}
Thus,
\begin{equation*}
    \phi_p: \mc G_n \longrightarrow \mc D_n^{\leq}
\end{equation*}
is a channel whose output is a degree-count vector. 

We call one realization, $t \sim \phi_p(G)$, a degree-sequence trace which can be directly computed from the subgraph trace $\Phi_p(G)$. 
Given $m$ independent channel outputs, we define our observation model to be
\begin{equation*}
    t^{(1)},\ldots,t^{(m)} \overset{\mathrm{i.i.d.}}{\sim} \phi_p(G),
\end{equation*}
so that an algorithm receives the vectors $t^{(1)}, \dots, t^{(m)}$, rather than the full topological information of each subgraph. 
We note that all results in this work also hold for the problem of estimating the degree sequence from $m$ observations of $\Phi_p(G)$. 

Our goal is to reconstruct the original degree-count vector $d_G$ from these traces. An estimator for this problem is therefore a map 
\begin{equation*}
    \widehat d_m:\left(\mc D_n^{\leq}\right)^m\longrightarrow\mc D_n.
\end{equation*}

% For $\delta\in(0,\frac12)$, define the degree-sequence trace
% complexity by
% \[
%     m_{\mathrm{deg}}(n,p,\delta)
%     :=
%     \min_{\hat d_m}
%     \left\{
%         m: \sup_{G\in\mc G_n}
%         \Prob_G
%         \left[
%             \widehat d_m(t_1,\ldots,t_m)
%             \neq
%             d_G
%         \right]
%         \leq\delta
%     \right\}.
% \]

The complete distribution of $\phi_p(G)$ may depend on the topology of $G$ beyond its degree-count vector $d_G$. Consequently, it would not generally be correct to write the channel as $\phi_p(d_G)$.
Its expectation, however, depends only on $d_G$.

Define
\begin{equation*}
    v_G := \E[\phi_p(G)].
\end{equation*}
For every $i\in\{0,\ldots,n-1\}$, we have 
\begin{equation}
    \label{eq:expected-trace-degree-count}
    (v_G)_i = q^{i+1} \sum_{\ell=i}^{n-1} \binom{\ell}{i} p^{\ell-i}d_{G,\ell}.
\end{equation}

Indeed, consider an original vertex of degree $\ell$. For this vertex
to have degree $i$ after deletion, the vertex itself must be retained,
exactly $i$ of its $\ell$ neighbors must be retained, and its remaining
$\ell-i$ neighbors must be deleted. The probability of this event is
\[
    q
    \binom{\ell}{i}
    q^i
    p^{\ell-i}
    =
    q^{i+1}
    \binom{\ell}{i}
    p^{\ell-i}.
\]
Summing over the $d_{G,\ell}$ original degree-$\ell$ vertices and then
over all $\ell\geq i$ yields
\eqref{eq:expected-trace-degree-count}.

% Equation~\eqref{eq:expected-trace-degree-count} is triangular in the
% coordinates of $d_G$. The coefficient of $d_{G,i}$ in $(v_G)_i$ is
% $q^{i+1}>0$. Therefore, for every $q>0$, the exact mean vector $v_G$
% uniquely determines $d_G$. The statistical difficulty is that $v_G$
% is unknown and must be estimated from finitely many independent
% degree-sequence traces.

\subsection{Degree Moments}
\label{subsec:degree-moments}

For $u=(u_0,\ldots,u_{n-1})\in\mb R^n$, define
\[
    f_r(u)
    :=
    \sum_{i=0}^{n-1}
    i^r u_i,
    \qquad
    r\geq0,
\]
where $i^0:=1$. For the original degree-count vector,
\[
    f_0(d_G)=n
    \qquad\text{and}\qquad
    f_1(d_G)=2|E(G)|.
\]
We refer to $f_r(d_G)$ as the $r$th moment of the degree sequence.

Since $f_r$ is linear,
\begin{equation}
    \label{eq:firstTwoMomentEqs}
    f_r(v_G) = f_r\left(\E[\phi_p(G)]\right) = \E\left[f_r(\phi_p(G))\right].
\end{equation}
Thus, the moments of $v_G$ can be estimated using only the observed
degree-count vectors $t_1,\ldots,t_m$.

Our reconstruction method estimates a collection of low-order moments
of $d_G$ from the degree-sequence traces and then uses these moments
to identify the complete degree-count vector. The two main questions
are therefore how many moments are sufficient to determine $d_G$ and
how many degree-sequence traces are required to recover those moments
exactly.

Let $k_n^{\rm{Graph}}$ and $k_n^{\mc H}$ denote the minimum number of moments required to distinguish any degree sequence in $\mc D_n$ and $\mc H$ respectively, so that  
\begin{align}
    k_n^{\rm{Graph}} &:= \min\{k : \mc M_k(d) \text{ is injective on } \mc D_n\} \label{align:k_nG_Def} \\
    k_n^{\mc H} &:= \min\{k : \mc M_k(d) \text{ is injective on } \mc H_n\},
\end{align}
where $\mc M_k(d) := \left(f_0(d), \dots, f_{k-1}(d)\right)$ is the moment vector of $d$.

%%%%%%%%%%%%%%%%%%%%%%%%%%%%%%%
\section{Edge Count Reconstruction}
\label{sec:EdgeCountRecon}

As a warm-up to the full problem, we consider the problem of recovering the size (number of edges) of a graph from traces. We show that $O(n^3)$ traces are sufficient, and $\Omega(n^2)$ traces are required for any algorithm to recover the size of $G$ with high probability.  

Let
\[
    t^{(1)},\dots,t^{(m)} \overset{\mathrm{i.i.d.}}{\sim} \phi_p(G)
\]
be the observed degree-sequence traces. For each $j=1,\dots,m$,
define
\begin{equation}
    \label{eq:observed-trace-edge-count}
    X_j := \frac{1}{2} \sum_{i=0}^{n-1} i\,t^{(j)}_{i} = \frac{1}{2} f_1(t^{(j)}),
\end{equation}
to be the number of edges in trace $t_j$.  

Every original edge survives the vertex-deletion process precisely when both of its endpoints are retained. Since vertices are retained independently with probability $q$, each edge survives with probability $q^2$. Consequently,
\begin{equation}
    \label{eq:expected-trace-edge-count}
    \E[X_j] = q^2|E(G)|.
\end{equation}
Define the empirical mean $\overline X_m := \frac{1}{m} \sum_{j=1}^{m} X_j$ and the estimator
\begin{equation}
    \label{eq:edge-count-estimator}
    \widehat E_m := \left\lceil \frac{\overline X_m}{q^2}\right\rfloor.
\end{equation}
Since $|E(G)|$ is an integer, the event
\begin{equation}
    \label{eq:edgeCount_SuffCond}
    \left|\overline X_m-q^2|E(G)|\right|<\frac{q^2}{2}
\end{equation}
implies $\widehat E_m=|E(G)|$.
Then achievability follows from concentration around the mean. 

Since $0 \leq X_j \leq \binom{n}{2}$, one could apply Hoeffding's inequality, and immediately get 
\begin{equation}
    \label{eq:edge-count-hoeffding}
    \Prob\left[\left|\overline X_m-q^2|E(G)|\right|\geq\frac{q^2}{2}\right]\leq 2\exp\left\{-\frac{m q^4}{2\binom{n}{2}^{2}}\right\}.
\end{equation}
This leads to a sample complexity of $m = O\left(\frac{n^4}{q^4}\log\left(\frac1\delta\right)\right)$ to recover the size with probability $> 1 - \delta$. The dependence on $n$ can be improved by considering the variance of the sizes. 

\begin{theorembox}{Size Reconstruction Upper Bound}{EdgeCountRec_Achievability}
    % \label{thm:EdgeCountRec_Achievability}
    Using $m = O\left(\frac{n^3}{q^2} \log(1/\delta)\right)$ degree-sequences traces are sufficient to recover the size of $G$ with probability $> 1- \delta$. 
\end{theorembox}
\begin{proof}
    For each edge $e=\{u,v\}\in E(G)$, define
    \[
        Z_e := \mathbbm 1 \left\{u\text{ and }v\text{ are both retained} \right\}.
    \]
    Then
    \[
        X_j = \sum_{e\in E(G)}Z_e.
    \]
    For every edge $e$, we have
    \begin{equation}
        \E[Z_e] = q^2 \quad \text{and} \quad \Var(Z_e) = q^2(1-q^2).
    \end{equation}
    If two distinct edges $e$ and $e'$ are vertex-disjoint, then $Z_e, Z_{e'}$ are independent random variables, and hence
    \[
        \Cov(Z_e,Z_{e'})=0.
    \]
    If $e$ and $e'$ share exactly one endpoint, then all three vertices
    belonging to the two edges must be retained for both edges to survive.
    Therefore,
    \[
        \E[Z_eZ_{e'}] = q^3,
    \]
    and consequently,
    \[
        \Cov(Z_e,Z_{e'}) = q^3-q^4.
    \]
    
    The number of unordered pairs of edges sharing a vertex $v$ is $\binom{\deg_G(v)}{2}$, so that
    \begin{equation}
        \label{eq:trace-edge-count-variance}
        \Var(X_j)
        =
        |E(G)|q^2(1-q^2)
        +
        2(q^3-q^4)
        \sum_{v\in V(G)}
        \binom{\deg_G(v)}{2}.
    \end{equation}
    Using
    \begin{equation*}
        |E(G)| \leq \binom{n}{2} \quad \text{,and} \quad \sum_{v\in V(G)} \binom{\deg_G(v)}{2} \leq n\binom{n-1}{2},
    \end{equation*}
    we obtain
    \begin{align}
        \Var(X_j) &\leq \binom{n}{2} q^2 + 2q^3 n\binom{n-1}{2}\\
        &\leq \frac{3}{2} q^2n^3. \label{eq:trace-edge-count-variance-upper}
    \end{align}
    Applying Bernstein's inequality, we have 
    \begin{align*}
        \Proba{\abs{q^{-2} \bar X_m - \abs{E(G)}} \geq \frac{1}{2}} &\leq 2 \Exp{-\frac{1}{4} \frac{m q^2}{3 n^3 + \frac{1}{3} n^2}}\\
        &\leq 2 \Exp{- \frac{m q^2}{14n^3}}
    \end{align*}
    Taking $m = 14\frac{\log(2/\delta) n^3}{q^2}$ traces are enough for \eqref{eq:edgeCount_SuffCond} to hold with probability at least $1-\delta$. 
\end{proof}

% \begin{theorembox}{Size Reconstruction Upper Bound}{EdgeCountRec_Achievability}
%     % \label{thm:EdgeCountRec_Achievability}
%     Using $m = O\left(\frac{n^3}{q^2} \log(1/\delta)\right)$ traces are sufficient to recover the size of $G$ with probability $> 1- \delta$. 
% \end{theorembox}

We also provide a lower bound on the trace complexity of recovering the size of $G$. 

\begin{theorembox}{Size Reconstruction Lower Bound}{LowBound_EdgeCount}
    % \label{thm:LowBound_EdgeCount}
    $m = \Omega\left(n^2 \log \left(\frac{1}{\delta}\right)\right)$ degree-sequence traces are necessary to recover the size of $G$ with probability $ > 1 - \delta$. 
\end{theorembox}
The lower bound is proved by selecting two bipartite graphs whose induced statistical distributions are close, despite having different sizes. We defer this proof to \Cref{sec_app:Proof_LowBound_EdgeCount}. Of note, the lower bound holds even for algorithms that observe traces from $\Phi_p(G)$.

%%%%%%%%%%%%%%%%%%%%%%%%%%%%%%%
\section{Degree Sequence Reconstruction: Binomial Mixtures}
\label{sec:DegSeqRecon_BinomMixtures}

Transitioning to the problem of recovering the degree sequence of a graph using traces, we describe a method, based on the results of \cite{krishnamurthy2019trace}, for recovering the degree sequence using $m = \exp\{\tilde O(n^{1/3})\}$. We use the shorthand, $B(a,q)$ to denote a random variable following a Binom$(a, q)$ distribution. 

For the reader's convenience, we state \cite[Theorem 8]{krishnamurthy2019trace} here,
\begin{lemmabox}{}{BinomMixtureSampComp}
    % \label{lem:BinomMixtureSampComp}
    Let $\mc M$ be a mixture of $d = \poly(n)$ binomial distributions, described by drawing a sample from $B(a_i, q)$ with probability $\alpha_i$, where $q = \Omega\left(\sqrt{\frac{\log n}{a}}\right) \in (0,1)$, $0 \leq a_1, \dots, a_d \leq a$, and each $\alpha_i \in [0,1]$ has polynomial precision. Then $\Exp{O((a/q)^{1/3})\log ^{2/3}n }$ samples suffice to learn the parameters of $\mc M$ exactly, with high probability. 
    
\end{lemmabox}

To recover the degree sequence we reduce the problem to parameter estimation by converting trace samples to samples from a binomial mixture model via rejection sampling. Specifically, the binomial mixture we sample from has weights $\alpha_i = \frac{d_i}{n}$, so that applying \Cref{lem:BinomMixtureSampComp} recovers the degree sequence.

\begin{theorembox}{}{binomMixtureTraceParamRecon}
    % \label{thm:binomMixtureTraceParamRecon}
    For sufficiently large $n$ there exists an algorithm to recover the degree sequence of a graph $G$ using 
    \begin{equation*}
        m = \frac{1}{q}\Exp{ O\left(\left(\frac{n}{q}\right)^{1/3} \log^{2/3}(n)\right)}
    \end{equation*}
    traces. 
\end{theorembox}

\begin{proof}
    We begin by describing the process of reducing the degree sequence problem to a problem of learning the parameters of a binomial mixture by rejection sampling. 

    \begin{tcolorbox}[
        colback=white,
        colframe=black,
        boxrule=0.8pt,
        arc=1pt,
        left=6pt,
        right=6pt,
        top=6pt,
        bottom=6pt
    ]
        \label{box:RejectionSampling}
        \textit{Rejection Sampling Procedure}: Given $t \in \Z_{\geq}^{n}$ a degree sequence trace
        \begin{enumerate}
            \item Accept $t$ with probability $\frac{\|t\|_1}{n}$, otherwise set $b := \perp$
            \item If we accept, take $b := i$ with probability $\frac{t_i}{\|t\|_1}$
        \end{enumerate}
        Then conditioned on accepting, $b \sim \sum_{i=0}^{n-1} \frac{d_i}{n} B(i,q)$. 
    \end{tcolorbox}

    Let $t = (t_0, \dots, t_{n-1})$ be an arbitrary trace of $G$. From \eqref{eq:expected-trace-degree-count}, we have 
    \begin{equation*}
        \E[t_r] = q \sum_{\ell = r}^{n-1} d_\ell \binom{\ell}{r} q^r p^{\ell - r}.
    \end{equation*}
    Take $\alpha_\ell = \frac{d_\ell}{n} \in [0,1]$. Let $M_d$ be a binomial mixture defined so that 
    \begin{align*}
        M_d(r) &:= \Proba{M_d = r}\\
        &= \sum_{\ell = 0}^{n-1} \alpha_\ell \Proba{B(\ell, q) = r}\\
        &= \sum_{\ell = r}^{n-1} \frac{d_\ell}{n} \binom{\ell}{r} q^r p^{\ell - r}\\
        = \frac{\E[t_r]}{n q}.
    \end{align*}
    
    Define $X$ to be a random variable resulting by setting $X = \perp$ with probability $1 - \frac{\oneNorm{t}}{n}$, and otherwise taking $X = i$ with probability $\frac{t_i}{\oneNorm{t}}$. 
    This is equivalent to accepting the sample with probability $\frac{\|t\|_1}{n}$, and letting $X$ be its degree.
    
    For $t \neq 0$, we have 
    \begin{equation*}
        \Proba{X = r| t} = \frac{\|t\|_1}{n} \cdot \frac{t_r}{\|t\|_1} = \frac{t_r}{n}
    \quad \text{and} \quad
        \Proba{X = \perp| t} = 1 - \frac{\|t\|_1}{n}.
    \end{equation*}
    Now, 
    \begin{align*}
        \Proba{X = r} &= \sum_{t_r \in \{0, \dots, n\}} \Proba{X = r | t_r} \Proba{t_r}\\
        &= \sum_{t_r} \frac{t_r}{n} \Proba{t_r} \;= \frac{\E[t_r]}{n}
        \;= q M_d(r).
    \end{align*}
    Furthermore, $\Proba{X \neq \perp} = q$, so that 
    \begin{equation*}
        \Proba{X = r | X \neq \perp} = M_d(r).
    \end{equation*}
    Therefore, 
    \begin{equation*}
        X | \{X \neq \perp\} \sim \sum_{\ell = 0}^{n-1} \frac{d_\ell}{n} B(\ell, q),
    \end{equation*}
    is a mixture of binomial distributions with mixture coefficients $\frac{d_\ell}{n}$. By performing rejection sampling in this way on each trace, we have probability $q$ of sampling directly from this binomial mixture. 

    Given a specified success probability, $\delta \in (0,1/2)$ and a desired number of binomial mixture samples, $M$, take 
    \begin{equation*}
        m = \left\lceil \frac{2M + 8 \log(1/\delta)}{q} \right\rceil,
    \end{equation*}
    traces. The number of accepted samples is $Y \sim \rm{Bin}(m, q)$, so that $\E[Y] = m q \geq 2M$. The binomial Chernoff bound gives us  
    \begin{align*}
        \Proba{Y < M} &\leq \Proba{Y < \frac{\E[Y]}{2}}\\
        &\leq e^{-\frac{mq}{8}}\\
        &\leq \delta.
    \end{align*}
    so that with probability $\geq 1 - \delta$ we recover $M$ samples from the binomial mixture. Thus, transferring from a problem of trace reconstruction to binomial mixture parameter estimation costs us $O\left(\frac{M + \log(1/\delta)}{q}\right)$ traces. Applying \Cref{lem:BinomMixtureSampComp} with $a = n$ recovers the parameters of the mixture. 
\end{proof}

While this method gives the best known trace complexity for recovering the degree sequence, we remark that \cite[Theorem 8]{krishnamurthy2019trace} gives only an information theoretic separation between possible binomial mixtures. We are unaware of any sub-exponential algorithms for recovering the mixture components, while an exhaustive search takes exponential time.

%%%%%%%%%%%%%%%%%%%%%%%%%%%%%%%
\section{Degree Sequence Reconstruction: Degree Moments}
\label{sec:DegSeqRecon_EdgeMoments}

We give a second algorithm to recover the degree sequence, utilizing the degree moments described in \Cref{subsec:degree-moments}.  
The result is split into three parts: In the first we convert our statistical problem to an algebraic problem by exploiting a relationship between the post-channel degree moments of $\Phi_p(G)$ and the true degree moments of $G$. Next, we prove that $k = \tilde \Theta (\sqrt{n})$ degree moments are necessary and sufficient to uniquely identify the degree sequence of $G$. Finally, we provide an algorithm to efficiently compute the degree sequence of $G$ using its degree moments.

%%%%%%%%%%%%%%%%%%%%%%%%%%%%%%%
\subsection{From traces to degree moments}
\label{subsec:EdgeMoment_StatisticsToAlgebra}
% For now, assume that you are given $k \in \mb N$, the number of moments required to reconstruct the degree sequence of $G$. Let $v_G = \E[\phi_p(G)]$, and suppose that you know $f_0(v_G), \dots, f_{k-1}(v_G)$ exactly. 

We begin by relating the factorial moments of the degree-sequence traces to the factorial moments of the degree sequence of the original graph $G$. 

For $u \in \R^{n}$, let 
\begin{equation}
    \label{eq:fallingFactorialMoment}
    h_\ell(u) := \sum_{r=0}^{n-1} r^{\underline\ell} u_r, 
\end{equation}
be the $\ell^{\text{th}}$ falling factorial moment of $u$, where $r^{\underline{\ell}} = r(r-1) \dots (r-\ell + 1)$ denotes the falling factorial power. 
% Define $w_\ell := h_\ell(d_G)$. 
We have the following identity relating $h_\ell(v_G)$ and $h_\ell(d_G)$,
\begin{propositionbox}{}{FallingFactorialMomentScaling}
    % \label{prop:FallingFactorialMomentScaling}
    For every $\ell \in \mb N$,
    \begin{equation}
        \label{eq:FallingFactorialMomentScaling}
        h_\ell(v_G) = q^{\ell+1} h_\ell(d_G).
    \end{equation}
\end{propositionbox}
\begin{proof}
    Let $t = \phi_p(G)$ be a degree sequence trace. For every $v \in V(G)$, let $I_v$ be a binary random variable indicating if $v$ survived the deletion channel. Suppose $v$ has degree $r$ in $G$, and take 
    \begin{equation*}
        X_v := \sum_{u \in N_G(v)} I_u,
    \end{equation*}
    where $N_G(v)$ denotes the vertex neighborhood of $v$. Then $X_v \sim \rm{Bin}(r, q)$. Note that, $I_v$ and $X_v$ are independent, so that 
    \begin{align*}
        \E[I_v X_v^{\underline{\ell}}] &= \E[I_v] \E[X_v^{\underline{\ell}}]\\
        &= q^{\ell + 1} r^{\underline{\ell}}. 
    \end{align*}
    In the above expression we used a fact about the falling factorial moment \cite{Potts1953FactorialMoments},  $\E[X^{\underline{\ell}}] = r^{\underline{\ell}} q^\ell$. 

    Then we have, 
    \begin{align*}
        h_\ell(t) &= \sum_{v \in V(G)} I_v X^{\underline{\ell}}_v,
    \end{align*}
    and since there are $d_{G,r}$ vertices in $G$ with degree $r$,
    \begin{align*}
        h_\ell(v_G) &= h_\ell(\E[t]) = \E[h_\ell(t)]\\
        &= \sum_{r=0}^{n-1} d_{G,r} q^{\ell+1} r^{\underline{\ell}} 
        = q^{\ell+1} h_\ell(d_G).
    \end{align*}
\end{proof}
Now we will need some results about polynomial change of basis. Specifically, let 
\begin{equation*}
    \mathbf{a}^{(k)}(x) = \begin{bmatrix}
        1\\
        x\\
        \vdots\\
        x^{k-1}
    \end{bmatrix}, \qquad \mathbf{b}^{(k)}(x) = \begin{bmatrix}
        1\\
        x^{\underline{1}}\\
        \vdots\\
        x^{\underline{k-1}}
    \end{bmatrix},
\end{equation*}
be polynomials expressed in the standard basis and falling factorial basis respectively. Take $s(r,j)$ and $\bfrac{r}{j}$ to be as in \eqref{eq:StirlingNumbersFirstKind} and \eqref{eq:StirlingNumbersSecondKind} respectively. 
Define 
\begin{equation*}
    \mathbf{S}_k := [s(r,j)]_{r,j=0}^{k-1}, \qquad \mathbf{T}_k := \left[\bfrac{r}{j}\right]_{r,j=0}^{k-1}.
\end{equation*}
Equations \eqref{eq:StirlingNumbersFirstKind} and \eqref{eq:StirlingNumbersSecondKind} are expressed linear algebraically as
\begin{equation*}
    \mathbf{b}^{(k)}(x) = \mathbf{S}_k \mathbf{a}^{(k)}(x), \qquad \mathbf{a}^{(k)}(x) = \mathbf{T}_k \mathbf{b}^{(k)}(x).
\end{equation*}
Furthermore, both $\mathbf{S}_k, \mathbf{T}_k$ are invertible, and  $\mathbf{S}_k = \mathbf{T}_k^{-1}$. 

\noindent For $u \in \R^{n}$, take 
\begin{equation*}
    \mathbf{m}^{(k)}(u) := \begin{bmatrix}
        \sum_{r} u_r\\
        \sum_r r u_r\\
        \vdots\\
        \sum_r r^{k-1} u_r
    \end{bmatrix} = \begin{bmatrix}
        f_0(u)\\
        f_1(u)\\
        \vdots\\
        f_{k-1}(u)
    \end{bmatrix}, \qquad \mathbf{h}^{(k)}(u) := \begin{bmatrix}
        \sum_r u_r\\
        \sum_r r^{\underline{1}} u_r\\
        \vdots\\
        \sum_{r} r^{\underline{k-1}} u_r
    \end{bmatrix} = \begin{bmatrix}
        h_0(u)\\
        h_1(u)\\
        \vdots\\
        h_{k-1}(u)
    \end{bmatrix}. 
\end{equation*}
Here $\mathbf{m}_k(u)$ denotes the vector of the first $k$ degree moments of $u$, and $\mathbf{h}_k(u)$ the vector of the first $k$ falling factorial moments of $u$. 
We have 
\begin{equation}
    \label{eq:ChangeOfBasisFallingFacAndStandard}
    \mathbf{h}^{(k)}(u) = \mathbf{S}_k \mathbf{m}^{(k)}(u), \qquad \mathbf{m}^{(k)}(u) = \mathbf{T}_{k} \mathbf{h}^{(k)}(u).
\end{equation} 

\noindent Now, take
\begin{align*}
    \mathbf{h}_G &:= \begin{bmatrix}
        h_0(d_G)  & h_1(d_G) & \dots & h_{k-1}(d_G)
    \end{bmatrix}^T\\
    \mathbf{h}_v &:= \begin{bmatrix}
        h_0(v_G) & h_1(v_G) & \dots & h_{k-1}(v_G)
    \end{bmatrix}^T\\
    \mathbf{m}_G &:= \begin{bmatrix}
        f_0(d_G) & f_1(d_G) & \dots & f_{k-1}(d_G)
    \end{bmatrix}^T\\
    \mathbf{m}_v &:= \begin{bmatrix}
        f_0(v_G) & f_1(v_G) & \dots & f_{k-1}(v_G)
    \end{bmatrix}^T. 
\end{align*}
Then for $D_q = \diag(q, q^2, \dots, q^{k})$, \Cref{prop:FallingFactorialMomentScaling} tells us that  
\begin{equation*}
    \mathbf{h}_v = D_q \mathbf{h}_G,
\end{equation*}
while \eqref{eq:ChangeOfBasisFallingFacAndStandard} gives us 
\begin{equation}
    \label{eq:conversionFromNoisyFallingFactorialToTrueEdgeMoments}
    % \mathbf{m}_G = \mathbf{S}_k D_q^{-1} \mathbf{S}_{k}^{-1} \mathbf{m}_v.
    \mathbf{m}_G = \mathbf{T}_k D_q^{-1} \mathbf{h}_v.
\end{equation}
Therefore, there is a bijection between the post-channel falling factorial degree moments and true degree moments of $G$. 

The final step of this section is to determine the trace complexity of estimating the first $k_n$ post-channel degree moments of $G$. 
\begin{theorembox}{}{TraceComplexity_EdgeMomentEstimation}
    % \label{thm:TraceComplexity_EdgeMomentEstimation}

    Let $k_n = o(n)$. The first $k_n$ degree moments of $G$ can be estimated exactly with probability $\geq 1 - \delta$ using 
    \begin{align*}
        m &= \log(2k_n/\delta)\Exp{(2k_n + 2) \left[\log\frac{2n}{q}\right]}\\
        &=\Exp{\tilde O (k_n)} \log\left(\frac{k_n}{\delta}\right)
    \end{align*}
    traces. 
    
\end{theorembox}

\begin{proof}
    % Let $k \equiv k_n$.

    % Recall that the $j^{\text{th}}$ falling factorial moment of $u \in \R^{n}$ is defined as
    % \begin{equation*}
    %     h_j(u) := \sum_{r=0}^{n-1} r^{\underline{j}} u_r, 
    % \end{equation*}
    
    By \eqref{eq:conversionFromNoisyFallingFactorialToTrueEdgeMoments}, it is enough to estimate the falling factorial moments with high probability.  
    Let $\bar \eta_\ell = \frac{1}{m} \sum_{i=1}^m h_\ell(t^{(i)})$ be the $\ell^{\text{th}}$ moment sample mean, and let our estimator for $h_\ell(d_G)$ be 
    \begin{equation*}
        \hat h_\ell = \left\lceil  \frac{\bar \eta_\ell}{q^{\ell + 1}} \right\rfloor.
    \end{equation*}
    Note that $0 \leq h_\ell(t^{(i)}) \leq n^{\ell+1}$. To recover $h_\ell(d_G)$ exactly, we require $\abs{\hat h_\ell - h_\ell(d_G)} \leq \frac{1}{2}$, as $h_\ell(d_G) \in \mb N$.   
    
    From \Cref{prop:FallingFactorialMomentScaling}, we have $\E[\bar \eta_\ell] = q^{\ell + 1 } h_\ell(d_G)$, so that  Hoeffding's inequality gives us 
    \begin{align*}
        \Proba{\abs{\frac{\bar \eta_\ell}{q^{\ell+1}} - h_\ell(d_G)} \geq \frac{1}{2}} &= \Proba{\abs{\bar \eta_\ell - h_\ell(v_G)} \geq \frac{q^{\ell+1}}{2}}\\
        &\leq 2 \Exp{- \frac{q^{2\ell+2} m }{2n^{2\ell+2}}}.
    \end{align*}
    Taking 
    \begin{align*}
        m 
        &= 2 \log\left(\frac{2}{\delta}\right) \Exp{(2\ell + 2) \left[\log\frac{2n}{q}\right]}
        % &= \Exp{\tilde O\left(\ell\right)} \log(1/\delta),
    \end{align*}
    we get 
    \begin{equation*}
        \Proba{\hat h_\ell = h_\ell(d_G)} \geq 1 - \delta. 
    \end{equation*}
    The final result follows from a union bound over all $k_n = O(n)$ falling factorial degree moments. 
\end{proof}

The implementation of this procedure is described in \Cref{alg:MomentEst}.  \Cref{thm:TraceComplexity_EdgeMomentEstimation} means that for $m = \Exp{\tilde O (k)} \log\left(\frac{1}{\delta}\right)$, the output of the algorithm matches the true graph moments, i.e. $\hat f_k = f_k(d_G)$, with probability $> 1-\delta$. The runtime of this algorithm is $O(m \poly(n,k))$.

\begin{algorithm}
    \caption{EstimateMoments($k, [t^{(1)}, \dots, t^{(m)}]$)}
    \label{alg:MomentEst}
    \begin{algorithmic}[1]
    
    \FOR{$\ell \in \{0, \dots, k-1\}$}
        \STATE $\bar \eta_\ell \gets \frac{1}{m} \sum_{i=1}^m h_\ell(t^{(i)})$
    
        \STATE $\hat h_\ell \gets \left\lceil \frac{\bar \eta_\ell}{q^{\ell+1}} \right\rfloor$
    \ENDFOR
    \STATE $\hat{\mathbf{w}}_G \gets \begin{bmatrix}
      \hat h_0 &
      \hat h_1 &
      \ldots &
      \hat h_{k-1}
    \end{bmatrix}^T$
    \STATE $\hat{\mathbf{f}} \gets \begin{bmatrix}
      \hat f_{0} &
      \hat f_{1}&
      \ldots &
      \hat f_{k-1}
    \end{bmatrix}^T := T_k \hat{\mathbf{w}}_G$
    % \STATE 
    \RETURN $\hat{\mathbf{f}}$.
    \end{algorithmic}
\end{algorithm}

%%%%%%%%%%%%%%%%%%%%%%%%%%%%%%%
\subsection{How Many Moments Do We Need?}
\label{subsec:EdgeMoment_NumberOfRequiredMoments}

% Previously, we assumed that we knew the number of moments required to uniquely reconstruct the degree sequence of a graph, and that it is sub-linear in $n$. For a given $n$, let $k$ denote this number. 
Let $k$ denote the number of degree-sequence moments required to uniquely reconstruct the degree sequence of a graph. 
In this section we will show that $k = \tilde \Theta(\sqrt n)$ moments are necessary and sufficient to recover degree sequence uniquely. 
To prove this, we convert the reconstruction problem to a problem of finding polynomials with a high multiplicity root at $x=1$ while satisfying an $L_1$ constraint on their coefficients. Using tight bounds for polynomials of this type  \cite{erdelyi2016coppersmith} we prove an upper bound of $k=O(\sqrt{n \log n})$. 
Furthermore, we also construct a polynomial which corresponds to a pair of graphical degree sequences whose first $k=\Omega\left(\sqrt{\frac{n}{\log n}}\right)$ degree moments match.   

Recall that $\mc D_n$ denotes the set of graphical degree sequences, and $\mc H_n$ denotes the set of relaxed degree sequences,
\begin{align*}
    \mc D_n = \{ d_G: G \in \mc G_n\}, \qquad \mc H_n = \left\{ d \in \Z_{\geq 0}^n: \sum_{i=0}^{n-1} d_i = n \right\}.
\end{align*}

Let $d^{(1)} \neq d^{(2)}$ denote two vectors in $\mc D_n$, so that 
\begin{align*}
    f_\ell(d^{(1)}) = f_\ell(d^{(2)}), \quad \forall \ell \in \{0,\dots, k-1\},
\end{align*}
i.e. they have the same first $k$ moments. Since the moments are linear functionals, taking $u = d^{(1)} - d^{(2)} \neq 0$,
\begin{equation*}
    f_\ell(u) = 0, \quad \forall \ell \in \{0, \dots, k-1\}.
\end{equation*}
Notice that $\oneNorm{u} \leq 2 n$. 

Now, let $\mathbf{A}_{n,k} \in \Z^{k \times n}$ be a Vandermonde matrix given by 
\begin{equation*}
    \mathbf{A}_{n,k} = [x^y]_{x=0,y=0}^{k-1, n-1} = \begin{bmatrix}
        1 & 1 & 1 & \dots & 1\\
        0 & 1 & 2 & \dots & n-1\\
        0^2 & 1^2 & 2^2 & \dots & (n-1)^2\\
        \vdots & \vdots & \vdots & \ddots & \vdots\\
        0^{k-1} & 1^{k-1} & 2^{k-1} & \dots & (n-1)^{k-1}
    \end{bmatrix}.
\end{equation*}
Then 
\begin{equation*}
    A_{n,k} u = \begin{bmatrix}
        f_0(u)\\
        f_1(u)\\
        \vdots\\
        f_{k-1}(u)
    \end{bmatrix} = 0.
\end{equation*}
Thus, if two degree sequences have matching first $k$ moments, then their difference is in the null space of $A_{n,k}$. 

We consider the relaxed class of degree sequences $\mc H_n$, and we want to understand when the moment map is injective on $\mc H_n$. Define the polynomial 
\begin{equation*}
    P_u(x) := \sum_{r=0}^{n-1} u_r x^r.
\end{equation*}
The next result shows that $u \in \rm{ker}(\mathbf{A}_{n,k})$ is equivalent to a divisibility condition on $P_u$. 
\begin{propositionbox}{}{KernelEquivToRootsOfUnity}
    % \label{prop:KernelEquivToRootsOfUnity}
    
    For all $u \in \Z^n$ and $1 \leq k \leq n$, 
    \begin{equation*}
        A_{n,k} u = 0 \quad \iff \quad (x-1)^k | P_u(x).
    \end{equation*}
    % i.e. $x=1$ is a root of $P_u(x)$ with multiplicity at least $k$. 
\end{propositionbox}
Here $P_2(x) | P_1(x)$ denotes  $P_2(x)$ dividing $P_1(x)$ evenly, so that $P_1(x) = P_2(x) P_3(x)$, where $P_1, P_2, P_3$ are polynomials. Note that $(x-1)^k | P_u(x)$ is equivalent to stating that $x=1$ is a root of $P_u(x)$ with multiplicity at least $k$. 

\begin{proof}
    The ordinary powers, 
    \begin{equation*}
        1, x, \dots, x^{k-1}
    \end{equation*}
    and the falling factorial powers
    \begin{equation*}
        1, x^{\underline{1}}, \dots, x^{\underline{k-1}}
    \end{equation*}
    are two bases for the space of polynomials of degree at most $k-1$.  
    
    For $u = (u_0, \dots, u_{n-1})$, define the linear functional 
    \begin{equation*}
        L_u(g) = \sum_{i=0}^{n-1} g(i) u_i,
    \end{equation*}
    on the space of polynomials $g$. 

    $(\implies)$ Suppose that $A_{n,k} u = 0$, which means that 
    \begin{equation*}
        L_u(x^j) = \sum_{i=0}^{n-1} i^j u_i = 0, \quad \forall 0, \dots, k-1
    \end{equation*}
    Since $1, x, \dots, x^{k-1}$ form a basis for all polynomials of degree at most $k-1$, by linearity 
    \begin{equation*}
        L_u(g) = 0
    \end{equation*}
    for every such polynomial, $g_r$. In particular, taking $g_r(x) = x^{\underline{r}}$, we have 
    \begin{equation*}
        L_u(g_r) = \sum_{i=0}^{n-1} i^{\underline{r}} u_i = 0, \quad \forall r \in \{0, \dots, k-1\}.
    \end{equation*}
    Differentiating $P_u(x)$ $r$-times, and evaluating at $1$ gives
    \begin{equation*}
        P_u^{(r)}(1) = \sum_{i=0}^{n-1} i^{\underline{r}} u_i = 0,
    \end{equation*}
    so that $x=1$ is a zero of multiplicity at least $k$. 

    $(\impliedby)$ The converse follows a similar process to the forward direction, only this time using the fact that $\{x^{\underline{r}}\}_{r=0}^{k-1}$ forms a basis for the set of polynomials of degree $k-1$, including $g_r(x) = x^r$. Starting from $P_u^{(r)}(1) = 0$, $\forall r \in \{0, \dots, k-1\}$, we get that $L_u(x^{\underline{r}}) = 0$, and therefore $L_u(g_r) = 0$. 
\end{proof}

Taking $u$ to be the difference between two relaxed degree sequences, then agreeing on the first $k$ moments is equivalent to
\begin{equation}
    \label{eq:ConditionsOnPolynomial}
    (x-1)^k | P_u(x), \quad \text{ and } \quad \oneNorm{u} \leq 2n.
\end{equation}
Then finding the number of moments required to reconstruct the degree sequence is equivalent to bounding the largest possible multiplicity of $x=1$ of polynomials with bounded coefficients.

Using this transformation, we will show the following result,
\begin{theorembox}{}{GraphicalKnTheta}
    Let $k_n^{\rm{Graph}}$ be as in \eqref{align:k_nG_Def}. Then  
    \begin{equation*}
        k_n^{\rm{Graph}} = \tilde\Theta\left(\sqrt n\right)
    \end{equation*}
\end{theorembox}

\subsubsection{Upper Bound}
\label{subsubsec:UpperBound}

\begin{remarkbox}{}{RelationToScottsWork}
    In our notation, applying \cite[Lemma 11]{Groenland2026SmallerCards} to the multisets represented by vectors in $\mc H_n$ gives
    \begin{equation*}
        k_n^{\rm{Graph}} \leq k_n^{\mc H} \leq \left\lfloor \sqrt{2(n-1) \log(2n)} \ \right \rfloor + 1,
    \end{equation*}
    which improves the upper bound in \Cref{thm:upperBound_kn}. 
    The authors became aware of this after the initial version of this manuscript appeared. We retain the arguments below for completeness. 
\end{remarkbox}

We will need a result in \cite{erdelyi2016coppersmith}, so we first give an overview of the notation. 
Take $N \in \N$, and $\alpha \in (0, 1/2]$ to be constants. Let $\kappa_1(N, \alpha)$ be the largest possible value of $k$ so that $\exists Q \not \equiv 0$ a polynomial with complex coefficients of the form 
\begin{equation}
    \label{eq:KappaRequirements}
    Q(x)= \sum_{j=0}^{N} a_j x^j, \quad \abs{a_0} \geq \alpha \sum_{j=1}^{N} \abs{a_j},
\end{equation}
such that $(x-1)^k|Q(x)$. Let $\mu_\infty(N, \alpha)$ be the smallest value of $k$ for which there is a polynomial of degree $k$ with complex coefficients such that 
\begin{equation}
    \label{eq:MuRequirements}
    \abs{Q(0)} > \frac{1}{\alpha} \max_{1 \leq j \leq N} \abs{Q(j)}.
\end{equation}
The following is a specialized version of \cite[Theorem 2.2]{erdelyi2016coppersmith} 
\begin{lemmabox}{}{PolyBound_ExternalResult}
    % \label{lem:PolyBound_ExternalResult}
    Let $\kappa_1$ be the largest value of $k$ so that $\exists Q \not \equiv 0$ a polynomial with complex coefficients satisfying \eqref{eq:KappaRequirements}, and take $\mu_\infty$ to be smallest value of $k$ so that there exists a polynomial $Q$ satisfying \eqref{eq:MuRequirements}. Then 
    \begin{equation*}
        \frac{2}{7} \min\left\{ \sqrt{N \log \frac{1}{\alpha}}, N \right\} \leq \kappa_1(N, \alpha) \leq \mu_\infty(N, \alpha) \leq 13 \min\left\{\sqrt{N \log \frac{1}{\alpha}}, N\right\} + 4.
    \end{equation*} 
\end{lemmabox}
Note that \cite[Theorem 2.2]{erdelyi2016coppersmith} includes a separate bound when $\alpha \in (1/2, 1]$, which is unnecessary for this work.  

Recall that $k_n^{\mc H}$ is the minimum number of degree moments required to distinguish any generalized degree sequence in $\mc H$. Now, we are prepared to state our upper bound on the number of required moments. 
\begin{theorembox}{}{upperBound_kn}
    % \label{thm:upperBound_kn}
    For every $n \geq 2$, we have 
    \begin{equation*}
        k_n^{\mc H} \leq 13 \sqrt{(n-1) \log(2n)} + 5.
    \end{equation*}
\end{theorembox}
\begin{proof}
    As a consequence of \Cref{lem:PolyBound_ExternalResult}, any polynomial
    \begin{equation*}
        Q(x) = a_0 + a_1 x + \dots + a_N x^N
    \end{equation*}
    satisfying
    \begin{equation*}
        \abs{a_0} \geq \alpha \sum_{j=1}^N \abs{a_j}, \quad 0 < \alpha \leq \frac{1}{2},
    \end{equation*}
    can have the multiplicity of the zero of $Q$ at $x=1$ at most
    \begin{equation}
        \label{eq:MaxNumRootOfUnity}
        13 \sqrt{N \log \frac{1}{\alpha}} + 4.
    \end{equation}

    For any $u = d^{(1)} - d^{(2)}$ for $d^{(1)}, d^{(2)} \in \mc H_n$, take $P_u(x) = \sum_{i=0}^{n-1} u_i x^i$. Let $s := \min\{i: u_i \neq 0\}$ and write $P_u(x) = x^s Q(x)$, for $Q(x) = a_0 + \dots a_N x^N$ and $a_0 = u_s \neq 0$. 
    Since 
    \begin{equation*}
        \sum_{j=1}^N \abs{a_j} \leq \oneNorm{u} \leq 2 n,
    \end{equation*}
    we have 
    \begin{equation*}
        \abs{a_0} \geq 1 \geq \frac{1}{2n} \sum_{j=1}^N \abs{a_j}.
    \end{equation*}
    Then the multiplicity of zeros in $Q(x)$ can be at most 
    \begin{equation*}
        13 \sqrt{(n-1) \log (2n)} + 4,
    \end{equation*}
    where we used \eqref{eq:MaxNumRootOfUnity} with $N \leq n-1$ and $\alpha = \frac{1}{2n}$.  
    Then due to \eqref{eq:ConditionsOnPolynomial}, taking 
    \begin{equation*}
        k = 13 \sqrt{(n-1) \log (2n)} + 5
    \end{equation*}
    gives us the result. 
\end{proof}

\subsubsection{Lower Bound}

Here we will show that $\Omega(\sqrt{n/\log n})$ moments are insufficient to distinguish all graphical degree sequences. To prove this, we construct a polynomial with coefficients in $\{-1, 0, 1\}$, which has many roots at $1$. We then show that the specific polynomial we select corresponds to a pair of graphical degree sequences using the Erdös-Gallai theorem. 

% Remember that 
% \begin{align*}
%     k_n^{\rm{Graph}} &:= \min\{k : \mc M_k(d) \text{ is injective on } \mc D_n\},\\
%     k_n^{\mc H} &:= \min\{k : \mc M_k(d) \text{ is injective on } \mc H_n\}
% \end{align*}
% with $\mc M_k(d) := \left(f_0(d), \dots, f_{k-1}(d)\right)$ the moment vector of $d$. 

Recall that we use the term \textit{degree sequence} when referring to a vector that describes the counts of vertices with a certain degree in a graph $G$, and we use the term \textit{ordered degree sequence} to refer to the ordered vector with each element being the degree of a particular vertex in $G$. 

We say that $a = (a_1, \dots, a_n) \in \Z^n$ is graphical \begin{equation}
    a_1 \geq a_2 \geq \dots \geq a_{n},
\end{equation}
and there exists a graph $G$ having the ordered degree sequence $a$. 

Now, we recollect the Erdös-Gallai theorem\cite{TRIPATHI_2003},
\begin{lemmabox}{}{ErdosGallai}
    % \label{lem:ErdosGallai}
    A sequence of positive integers $a = (a_1 \geq a_2 \geq \dots \geq a_n)$ is graphical if and only if $\sum_{i=1}^n a_i$ is even, and for every $t \in \{1, \dots, n-1\}$ we have 
    \begin{equation*}
        \sum_{i=1}^t  a_i \leq t(t-1) + \sum_{i=t+1}^n \min\{t, a_i\}
    \end{equation*}
\end{lemmabox}
We will use an alternative condition to certify whether a vector is graphical. 
\begin{propositionbox}{}{CertificateOfGraphicallity}
    % \label{prop:CertificateOfGraphicallity}
    Given $L \in \N$, let $a_1 \geq a_2 \geq \dots \geq a_{4L}$ be a sequence of integers satisfying 
    \begin{equation*}
        L \leq a_i \leq 2L, \quad \forall i \in \{1, \dots, 4L\}.
    \end{equation*}
    Then, if $\sum_{i=1}^{4L} a_i$ is even, $a = (a_1, \dots, a_{4L})$ is graphical. 
\end{propositionbox}
\begin{proof}
    By \Cref{lem:ErdosGallai}, we only need to show that 
    \begin{equation}
        \label{eq:ErdosGal_Expression}
        t(t-1) + \sum_{i=t+1}^{4L} \min\{t, a_i\} - \sum_{i=1}^{t} a_i \geq 0, \quad \forall t \in \{1, \dots, 4L\}.
    \end{equation}

    We can break this into two cases, when $t \leq L$ and when $t > L$. 

    For $t \leq L$, we have $\min\{t, a_i\} = t$. Then \eqref{eq:ErdosGal_Expression} can be bounded by 
    \begin{align*}
        t(t-1) + \sum_{i=t+1}^{4L} t - \sum_{i=1}^t a_i &= t(t-1) + t(4L - t) - \sum_{i=1}^{t}a_i\\
        &\geq t(t-1) + t(4L - t) - t2 L\\
        &= (2L-1) t
    \end{align*}

    For $t > L$, we have $\min\{t, a_i\} \geq  L$, so that 
    \begin{align*}
        t(t-1) + \sum_{i=t+1}^{4L} \min\{t, a_i\} - \sum_{i=1}^t a_i &\geq t(t-1) + L (4L - t) - t 2L\\
        &= t^2 - t(3L + 1) + 4L^2\\
        &\geq t^2 - t(4L) + 4L^2\\
        &= (t- 2L)^2.
    \end{align*}
\end{proof}

We are now prepared to prove our lower bound. 

\begin{theorembox}{}{lowerBound_kn}
    % \label{thm:lowerBound_kn}
    For all $n$ sufficiently large we have 
    \begin{equation}
        \label{eq:graphicalMoments_LowerBound}
        k_n^{\rm{Graph}} \geq \frac{1}{2} \sqrt{\frac{n}{\log_2(n+1)}} - 1
    \end{equation}
\end{theorembox}

\vspace{-1em}

\begin{proof}
    We show this result by finding two degree sequences, $d^{(1)}, d^{(2)}$ with matching first $\left\lfloor \frac{1}{2} \sqrt{\frac{n}{\log_2(n+1)}} \right\rfloor$ moments. 

    Without loss of generality, we assume that $n \mod 4 = 0$. To see this, note that we could take $\dot n = n - (n \mod 4)$. Since 
    \begin{align*}
        \frac{1}{2} \sqrt{\frac{\dot n}{\log_2(\dot n + 1)}} &\geq \frac{1}{2} \sqrt{\frac{n-3}{\log_2(n-2)}},
    \end{align*}
    and 
    \begin{equation*}
        \abs{\sqrt{\frac{n}{\log_2(n+1)}} -\sqrt{\frac{n-3}{\log_2(n-2)}}  } \to 0, \quad \text{ as } n \to \infty,
    \end{equation*}
    the difference between the two expressions can be made arbitrarily small, which can be accounted for by the $-1$ in \eqref{eq:graphicalMoments_LowerBound}. Then we could append isolated vertices to the graphs corresponding to $d^{(1)}, d^{(2)}$ until we had $n$ vertices. 

    In a similar manner, for notational convenience we assume that 
    \begin{equation*}
        \frac{1}{2} \sqrt{\frac{n}{\log_2(n+1)}}
    \end{equation*}
    is an integer in the following arguments, as the closest lower integer will suffice instead. Take $k := \frac{1}{2}\sqrt{\frac{n}{\log_2(n+1)}}$.

    Define $L := \frac{n}{4}$, and 
    \begin{equation*}
        A_L := \left\{ P(x) = 1 + \sum_{i=1}^{L-1} a_i x^i: a_i \in \{0,1\} \ \ \forall i   \right\}, \quad \abs{A_L} = 2^{L-1}.
    \end{equation*}
    For $P \in A_L$, let $S_k(P) = \left(P(1), P^{(1)}(1), \dots, P^{(k-1)}(1)\right)$ be the signature of $P$. There are at most $L^{\frac{k(k+1)}{2}}$ unique signatures for polynomials in $A_L$. Then we have 
    \begin{align*}
        L^{\frac{k(k+1)}{2}} &\leq L^{k^2}\\
        &= 2^{\frac{n}{4} \frac{\log_2(n/4)}{\log_2(n+1)}}\\
        &< 2^{\frac{n}{4} - 1}\\
        &= 2^{L -1}. 
    \end{align*}

    Therefore, there are two polynomials $P_1, P_2 \in A_L$, $P_1 \neq P_2$, with $S_k(P_1) = S_k(P_2)$. Take 
    \begin{equation*}
        Q(x) := P_1(x) - P_2(x) = \sum_{i=1}^{L-1} u_i x^i.
    \end{equation*}
    We have $Q(x) \neq 0$, and $Q^{(r)}(1) = 0$ for all $r \in \{0, \dots, k-1\}$, so that $(x-1)^k|Q(x)$. From \Cref{prop:KernelEquivToRootsOfUnity}, we have 
    \begin{equation*}
        \sum_{i=0}^{L-1} i^r u_i = 0, \quad \forall r \in \{0, \dots, k-1\}.
    \end{equation*}
    Since $f_0(u) = 0$, $\exists s \in \N$ such that $\|u\|_1 = 2 s \leq L$, which implies that 
    \begin{equation}
        s \leq \frac{L}{2} \leq \frac{n}{8}.
    \end{equation}
    Take $\tilde u \in \Z^n$ to be $u$ right shifted by $L$ terms, so that 
    \begin{align*}
        \tilde u_{L+i} := u_i, \quad \forall i \in \{0, \dots, L-1\},
    \end{align*}
    and $\tilde u_i = 0$ otherwise. Notice that $\tilde u$ has $f_r(\tilde u) = 0$ for every $r \in \{0, \dots, k-1\}$. 

    Decompose $\tilde u = \tilde u^+ - \tilde u^-$, where 
    \begin{equation*}
        \tilde u^+_i = \max\{\tilde u_i, 0\}, \quad \tilde u_i^- = \max\{- \tilde u_i, 0\}.
    \end{equation*}
    We have 
    \begin{equation*}
        \oneNorm{\tilde u^+} = \oneNorm{\tilde u^-} = s \leq \frac{n}{8}.
    \end{equation*}
    Let $e= f_1(\tilde u^+) \mod 2$, and $b \in \Z^n_{\geq 0}$ to be a common vector, with  
    \begin{align*}
        b_{2L} &:= n - s - e\\
        b_{2L-1} &:= e\\
        b_i &:= 0, \quad \forall i \in  \{0, \dots, n-1\} \setminus\{2L, 2L-1\}.
    \end{align*}
    
     Take $d^{(1)} := \tilde u^+ + b$, and $d^{(2)} := \tilde u^- + b$. Then $d^{(1)} - d^{(2)} = \tilde u$, so that the first $k$ moments match. Additionally, $\oneNorm{d^{(1)}} = \oneNorm{d^{(2)}} = n$, and $f_1(d^{(1)}),f_1(d^{(2)})$ are both even, because 
     \begin{align*}
         f_1(d^{(1)}) \mod 2 &= f_1(\tilde u^+ + b) \mod 2\\
         &= \left[f_1(\tilde u^+) + f_1(b)\right] \mod 2\\
         &= \left[f_1(\tilde u^+) + e (2L - 1)\right] \mod 2\\
         &= 0
     \end{align*}
     by design. 

     Furthermore, $d^{(1)}$, $d^{(2)}$ have only non-zero entries between $[L, 2L]$. Taking $a^{(1)}$, $a^{(2)}$ to be the corresponding ordered degree sequences, they satisfy the conditions of \Cref{prop:CertificateOfGraphicallity}. 
    
\end{proof}

Combining \Cref{thm:upperBound_kn} and  \Cref{thm:lowerBound_kn} we have 
\begin{equation*}
    \frac{1}{2}\sqrt{\frac{n}{\log_2(n+1)}} - 1 \leq k_n^{\rm{Graph}} \leq k_n^{\mc H} \leq 13 \sqrt{(n-1) \log(2n)} + 5,
\end{equation*}
which proves \Cref{thm:GraphicalKnTheta}
% \begin{equation}
%     \label{eq:graphical_kn_Theta}
%     k_n^{\rm{Graph}} = \tilde \Theta\left(\sqrt{n}\right).
% \end{equation}

%%%%%%%%%%%%%%%%%%%%%%%%%%%%%%%
\subsection{Degree Sequence from Degree Moments}
\label{subsec:EdgeMoment_Algorithm}

Here we give an efficient algorithm for recovering the degree sequence of a graph from true graph moments. As we showed in \Cref{thm:GraphicalKnTheta}, the vector $d \in \mc H_n$ is uniquely identified by its moments $f_0(d), \dots, f_{k_n^{\mc H}-1}(d)$, meaning that one should be able to recover the degree sequence from these moments. The goal of this section is to describe an algorithm that can efficiently solve for the degree sequence using moments. 

Unfortunately, this does not imply an efficient algorithm for recovering the exact degree sequence from traces, because we still need to distill the true moments from noisy observations (degree sequence traces). 
From \Cref{thm:TraceComplexity_EdgeMomentEstimation} and \Cref{thm:GraphicalKnTheta}, using  
\begin{equation*}
    m = \Exp{\tilde O(\sqrt{n})} \log\frac{1}{\delta},
\end{equation*}
traces, we can recover the first $k = O(\sqrt{n \log n})$ moments exactly, with runtime $O(m) = \Exp{\tilde{O}(\sqrt n)}$. This turns out to be the bottleneck of our  algorithm. 

While the sample complexity resulting from \Cref{sec:DegSeqRecon_BinomMixtures} is better than the moment based algorithm, we note that the result is purely information-theoretic. By this, we mean that the recovery algorithm is left implicit, and it is not known if the degree sequence can be recovered in sub-exponential time. Therefore, we find it useful to describe an efficient algorithm for recovering the degree sequence from moments. 

As a corollary of \Cref{lem:PolyBound_ExternalResult}, we can construct a polynomial which is 1 when evaluated at 0, but is very small when evaluated at other integers within a range.

Let $\R_k[x]$ denote the set of  polynomials with  real coefficients and degree at most $k$. 
\begin{corollarybox}{}{NicePolynomialsExist}
    % \label{cor:NicePolynomialsExist}
    For $k \geq 13 \sqrt{(n-1) \log 2n} + 5$, there exists a polynomial $Q(x) \in \mb{R}_{k-1}[x]$, satisfying
    \begin{equation*}
        \abs{Q(0)} = 1 > 2n \max_{1 \leq j \leq n-1} \abs{Q(j)}.
    \end{equation*}
\end{corollarybox}

\begin{proof}
    Take $\dot Q(x)$ a polynomial that achieves the upper bound on $\mu_\infty$  in \Cref{lem:PolyBound_ExternalResult}, normalized so that $\dot Q(0) = 1$, and take the coefficients of $Q$ to be the real parts of the coefficients of $\dot Q$. 
\end{proof}

\Cref{alg:RecoverDegreeSeqFromMoments} implements our method of recovering the degree sequence using the moments. The runtime of the algorithm is $\poly(nk)$. 

\begin{algorithm}
    \caption{Moment2Degree($n,[f_0, \dots, f_{k-1}]$)}
    \label{alg:RecoverDegreeSeqFromMoments}
    \begin{algorithmic}
        \STATE $Q(x) \gets \argmin\limits_{\substack{A(x) \in \mb{R}_{k-1}[x]\\ A(0) = 1}} \max_{1 \leq j \leq n-1} \abs{A(j)}$.
        \STATE $d \gets [0, \dots, 0]$.
        \STATE $M_r \gets f_r$, $\forall r \in \{0, \dots, k-1\}$.
        \FOR{$s \in \{0, \dots, n-1\}$}
            \STATE Expand $P(x) = Q(x-s) = \sum_{r=0}^{k-1} p_{s,r} x^r$
            \STATE $Z_s \gets \sum_{r=0}^{k-1} p_{s,r} M_r$
            \STATE $d[s] \gets \lceil Z_s \rfloor$. 
            \FOR{$r \in \{0, \dots, k-1\}$}
                \STATE $M_r \gets M_r - d[s]s^r$ 
            \ENDFOR
        \ENDFOR
        \RETURN $d$
    \end{algorithmic}
\end{algorithm}

\begin{theorembox}{}{RecDegSeq_FromMoments}
    % \label{thm:RecDegSeq_FromMoments}
    Let $f_0, \dots, f_{k-1}$ be the true degree moments of $G$, for $k \geq 13 \sqrt{(n-1) \log 2n} + 5$. Then \Cref{alg:RecoverDegreeSeqFromMoments} recovers $d_G$.
\end{theorembox}

\begin{proof}
    From \Cref{cor:NicePolynomialsExist}, solving the optimization problem 
    \begin{equation}
        \label{eq:polynomialOptimization}
        Q(x) = \argmin_{\substack{A(x) \in \mb{R}_{k-1}[x] \\ A(0) = 1}} \max_{1 \leq j \leq n-1} \abs{A(j)},
    \end{equation}

    yields a polynomial with 
    \begin{equation}
        \label{eq:IdealPolyBound}
        \frac{1}{2n} > \max_{1 \leq j \leq n-1} \abs{Q(j)}.
    \end{equation}

    Notice that 
    \begin{equation*}
        \tau = \min_{\substack{A(x) \in \mb R_{k-1}[x]\\ A(0) = 1}} \max_{1 \leq j \leq n-1} \abs{A(j)},
    \end{equation*}
    is equivalent to solving: 
    \begin{align*}
        \rm{minimize }\  &\tau\\
        \text{subject to: }\  &a_0 = 1\\
        &- \tau \leq \sum_{\ell = 0}^{k-1} j^\ell a_\ell \leq \tau, \quad \forall j \in \{1, \dots, n-1\},
    \end{align*}
    where $A(x) = \sum_{\ell = 0}^{k-1} a_\ell x^\ell$. Therefore, \eqref{eq:polynomialOptimization} can be formulated as a linear program, which can be solved in $\poly(n)$ time. 

    Let $q_0, \dots, q_{k-1}$ be the coefficients of $Q$, so that 
    \begin{equation*}
        Q(x) = \sum_{r=0}^{k-1} q_r x^r.
    \end{equation*}
    Let $f_0, \dots, f_{k-1}$ be the graph moments of $G$, so that 
    \begin{equation}
        f_\ell := f_\ell(d_G) = \sum_{i=0}^{n-1} d_i i^\ell. 
    \end{equation}
    Define $Z_0 := \sum_{r=0}^{k-1} q_r f_r$, then 
    \begin{align*}
        Z_0 &= \sum_{r=0}^{k-1} q_r\sum_{i=0}^{n-1} i^r d_i  = \sum_{i=0}^{n-1} d_i Q(i) = d_0 + \sum_{i=1}^{n-1} d_i Q(i).
    \end{align*}
    Therefore, 
    \begin{align*}
        \abs{Z_0 - d_0} &\leq \sum_{i=1}^{n-1} d_i\abs{Q(i)}\\
        &\leq \max_{1\leq j \leq n-1} \abs{Q(j)} \sum_{i=1}^{n-1} d_i < \frac{1}{2}.
    \end{align*}
    In the above result, we used the fact that $\sum_{i=0}^{n-1} d_i = n$ along with \eqref{eq:IdealPolyBound}. Then we have $\lceil Z_0\rfloor = d_0$.
    
    Now, assume for induction that we know $d_0, \dots, d_{s-1}$, and we want to estimate $d_s$. For every $r \in \{0, \dots, k-1\}$, define 
    \begin{equation*}
        M_r^{(s)} := f_r - \sum_{i=0}^{s-1} i^r d_i = \sum_{i=s}^{n-1} i^r d_i.
    \end{equation*}
    Additionally, let $P_s(x) = Q(x-s)$, so that $P_s(s) = 1$ and $\forall i > s$ we have 
    \begin{equation*}
        \abs{P_s(i) } < \frac{1}{2n}.
    \end{equation*}
    Let $p_{s,0}, \dots, p_{s,k-1}$ to be the coefficients of $P_s$, so that 
    \begin{equation*}
        P_s(x) = \sum_{r=0}^{k-1} p_{s,r} x^r.
    \end{equation*}
    Specifically, 
    \begin{equation*}
        p_{s,r} = \sum_{t=r}^{k-1} q_t \binom{t}{r} (-s)^{t-r}.
    \end{equation*}
    Our estimate $Z_s$ will be defined similar to $Z_0$, i.e. 
    \begin{align*}
        Z_s &:= \sum_{r=0}^{k-1} p_{s,r} M_r^{(s)} 
        = \sum_{i=s}^{n-1} d_i \sum_{r=0}^{k-1} p_{s,r} i^r\\
        &= \sum_{i=s}^{n-1} d_i P_s(i) 
        = d_s + \sum_{i=s+1}^{n-1} d_i P_s(i).
    \end{align*}
    Then 
    \begin{align*}
        \abs{Z_s - d_s} &\leq \sum_{i=s+1}^{n-1} d_i \abs{P_s(i)}\\
        &< \frac{1}{2},
    \end{align*}
    so that $\lceil Z_s \rfloor = d_s$. 
\end{proof}

The main algorithm for degree reconstruction from traces is described in \Cref{alg:DegRecon}, and the main result of this section is the following theorem. 
\begin{theorembox}{}{RecDegSeq_FromTraces}
    % \label{thm:RecDegSeq_FromTraces}
    Taking 
    \begin{align*}
        k &= \left\lceil 13 \sqrt{(n-1) \log(2n)} + 5\right\rceil\\
        m &= \log( k/\delta)\Exp{(2k + 2) \left[\log\frac{2n}{q}\right]}   \\
        &=\Exp{\tilde O (\sqrt{n})} \log\left(\frac{1}{\delta}\right),
    \end{align*}
    then \Cref{alg:DegRecon} recovers $d_G$ exactly with probability $ > 1-\delta$.
\end{theorembox}

\begin{proof}
    The proof follows directly from \Cref{thm:TraceComplexity_EdgeMomentEstimation} and   \Cref{thm:RecDegSeq_FromMoments}.
\end{proof}

\noindent The runtime of \Cref{alg:DegRecon} is 
\begin{equation}
    \label{eq:fullAlgRuntime}
    \poly(m) + \poly(n k) = \Exp{\tilde{O}(\sqrt{n})}.
\end{equation}

\begin{algorithm}
    \caption{ReconstructDegreeSequence(k, n, $[t^{(1)}, \dots, t^{(m)}]$)}
    \label{alg:DegRecon}
    \begin{algorithmic}[1]
    
      \STATE $\hat {\mathbf{f}} \gets \rm{EstimateMoments}(k, [t^{(1)}, \dots, t^{(m)}])$
      \STATE $\hat d_G \gets \rm{Moment2Degree}(n, \hat{\mathbf{f}})$
      \RETURN $\hat d_G$
    \end{algorithmic}
\end{algorithm}

%%%%%%%%%%%%%%%%%%%%%%%%%%%%%%%
\section{Conclusion}
\label{sec:conclusion}

We considered the problem of reconstructing the degree sequence of a graph from vertex deletion traces, and we provided two algorithms for this problem with complementary sample and runtime complexities. The first algorithm uses rejection sampling to sample from a binomial mixture distribution, where the mixture weights encode the degree sequence. By leveraging \cite[Theorem 8]{krishnamurthy2019trace}, we determine that the sample complexity of recovering the mixture parameters, and thus the degree sequence, is $\Exp{\tilde O(n^{1/3})}$. However, the only known decoding algorithm for this method involves brute force search, taking exponential time. The second algorithm makes use of a particular set of graph invariants, the degree moments, to recover the degree sequence using algebraic techniques. Utilizing a result on the multiplicity of roots at 1 for certain polynomials \cite{erdelyi2016coppersmith}, we were able to prove a tight bound up to logarithmic factors on the number of degree moments required for degree sequence reconstruction. This led to a sample complexity of $\Exp{\tilde O(n^{1/2})}$. By exploiting the same polynomial bound, we provide an algorithm to efficiently recover the degree sequence of a graph from its degree moments, leading to an overall runtime matching the sample complexity. As an introduction to the problem, we showed an upper and lower bound on the number of traces required to estimate the first moment (size) of a graph, which are within a factor of $n$ of each other. 

While the result of Sengupta and McGregor \cite{mcgregor2022graph} shows that graph trace reconstruction requires an exponential number of samples for arbitrary graphs, our result shows that for any pair of graphs with different degree sequences, a sub-exponential (yet still super-polynomial) number of samples suffices to distinguish. 

\noindent \textbf{Open Problems:} 
The most immediate open problem is to improve our trace complexity upper bound for degree-sequence reconstruction. The recent quasipolynomial result for string-trace reconstruction \cite{burudgunte2026quasipolynomial} suggests that exploiting joint information across traces is likely to improve the sample complexity. In contrast, our most effective method uses the information from a single vertex in each trace, which is seemingly wasteful. 

% Our best method reduces each accepted trace to the degree of a single uniformly random vertex, thereby ignoring correlations among vertices within the same trace. The recent quasipolynomial upper bound for binary-string trace reconstruction \cite{burudgunte2026quasipolynomial}, which overcomes barriers for local methods by exploiting global trace statistics, suggests that using joint information from many vertices in each graph trace could yield substantially better bounds.

% The most immediate extension to this work would be to improve the sample complexity upper bound for degree sequence reconstruction. We have several reasons to believe that our methods are far from the optimal trace complexity: The most direct piece of evidence is the recent result showing quasi-polynomial sample complexity for trace reconstruction \cite{burudgunte2026quasipolynomial}, which indicates that algorithms should take advantage of the global structure of the statistics to approach the optimal sample complexity. Notice that the best algorithm developed in this work uses the information from a single vertex in each trace. 

% Naturally, any improvement on the binomial mixture result of \cite[Theorem 8]{krishnamurthy2019trace} would improve on the bound described here. . Outside of this possibility, we believe that improving the the sample complexity for this problem would take new tools beyond those described in this work. Naturally, one might consider the fact that 
% We believe this will require new tools, or at min

There are several other open problems related to the graph trace reconstruction. One could also try to close the gap between the $O(n^3)$ upper and $\Omega(n^2)$ lower bounds for size reconstruction. The $\Omega(n^2)$ lower bound for size reconstruction is also the best lower bound for degree sequence reconstruction, so another direction would be to prove a non-trivial lower bound on degree sequence reconstruction. 
We have also left considerations of non-constant $q$ for future work. 
Additionally, one could consider channels beyond the vertex deletion channel, such as the edge addition and removal channel considered in \cite{McGregor2024GraphRecon}. 
Finally, it might be interesting to consider some variations of the trace reconstruction problem for unlabeled graphs, such as coded and approximate trace reconstruction. 

% The first is to understand several flavors of reconstruction, namely coded and approximate graph trace reconstruction. The coded trace reconstruction is of particular interest, as it could lead to techniques for encoding information on unlabeled graphs. Another direction to consider more general channel models, such as channels that add and remove edges independent of vertex deletions (as in \cite{McGregor2024GraphRecon}). 

\noindent \textbf{AI Disclosure}: OpenAI ChatGPT 5.6 Sol was used to assist with literature review, copy-editing, identifying the connection to \cite{krishnamurthy2019trace}, and compressing portions of \Cref{sec:introduction}, \Cref{sec:DegSeqRecon_BinomMixtures}, \Cref{sec:DegSeqRecon_EdgeMoments}, \Cref{sec:conclusion}, and \Cref{sec_app:Proof_LowBound_EdgeCount}. The authors independently re-derived and verified all mathematical claims and references and take full responsibility for the contents of the paper.

\ifAnon

\else
\section*{Acknowledgments}
The authors would like to thank Juliana Mini for contributions to the early stages of this work. 
We are grateful to Alex Scott for pointing out the connection to \cite{Groenland2026SmallerCards} and the resulting improvement to the moment-reconstruction upper bound. 
\fi

\bibliographystyle{alphaurl}
\bibliography{refs}

\clearpage
\appendix

\section{Proof of \texorpdfstring{ Theorem \ref{thm:LowBound_EdgeCount} }{Edge Count Lower Bound} }

\label{sec_app:Proof_LowBound_EdgeCount}

First we will need an intermediate result for the Hellinger distance between pairs of binomial random variables. For $P,Q$ probability distributions with support a subset of $\{0, \dots, d\}$, the Hellinger distance between $P$ and $Q$ is given by 
\[
    H^2(P,Q) = \frac{1}{2}\sum_{i=0}^d (\sqrt{p_i} - \sqrt{q_i})^2
\]

For $q \in (0, 1)$, let $\mu_A = \rm{Bin}(A, q)$ be the binomial distribution with $A$ trials and success probability $q$. Let $P =\mu_N \otimes \mu_N$ be the joint probability distribution of $(X,Y)$, where $X \sim \mu_N$ and $Y \sim \mu_N$, and take $F = \frac{1}{2} (\mu_{N-1} \otimes \mu_{N+1} + \mu_{N+1}\otimes \mu_{N-1})$ be the joint mixture distribution of a pair of binomial random variables. 

We will use the following intermediary result about the Hellinger distance between $P$ and $F$.
\begin{propositionbox}{Bounded Hellinger Distance}{interHellDist}
    % \label{prop:inter_HellDist}
    \[
        H^2(P,F) = O\left(N^{-2}\right).
    \]
\end{propositionbox}
\begin{proof}
    For $X \sim \mu_N$, let 
    \begin{align*}
        Z_X &:= X - \E[X]\\
        &= X - Nq\\
        U_X &:= \frac{Z_X}{N (1-q)}.
    \end{align*}
    Note that $\E[U_X] = 0$.
    For the following likelihood ratios, 
    \begin{align*}
        \frac{\mu_{N-1}(X)}{\mu_N(X)} &= \frac{N-X}{N(1-q)} = 1 - U_X\\
        \frac{\mu_{N+1}(X)}{\mu_N(X)} &= \frac{(N+1) (1-q)}{N+1-X} = 1 + U_X + \eps_X,
    \end{align*}
    where $\eps_X$ is the perturbation
    \[
        \eps_X := \frac{Z_X^2 - Z_X - N q(1-q)}{N (1-q) (N(1-q) + 1 - Z_X)}.
    \]

    For $(x,y) \in \{0, \dots, N\}^2$,  the likelihood ratio of $F$ over $P$ is
    \begin{align*}
        R(X,Y) &= \frac{F(X,Y)}{P(X,Y)}\\
        &= \frac{1}{2}  \frac{\mu_{N+1}(X) \mu_{N-1}(Y) + \mu_{N-1}(X) \mu_{N+1}(Y)}{\mu_N(X) \mu_N(Y)}\\
        % &= \frac{1}{2} [L_+(X) L_-(Y) + L_-(X) L_+(Y)]\\
        &= 1 - U_X U_Y + \frac{\eps_X + \eps_Y}{2} - \frac{U_X \eps_Y + U_Y \eps_X}{2}.
    \end{align*}

    Define the event $S := \{(X,Y) \in \left\{0, \dots, N\}^2\right\}$. We may bound the Hellinger distance using the likelihood ratio $R(X,Y)$ as follows
    \begin{align*}
        H^2(P, F) &= \frac{1}{2} \sum_{(x,y)} \left(\sqrt{P(x,y)} - \sqrt{F(x,y)}\right)^2 \\
        &= \frac{1}{2}\sum_{(x,y) \in S} \left(\sqrt{P(x,y)} - \sqrt{F(x,y)}\right)^2 + \frac{1}{2}\sum_{(x,y) \in S^C} F(x,y)\\
        &= \frac{1}{2} \sum_{(x,y) \in S} P(x,y) \left(1 - \sqrt{R(x,y)}\right)^2 + \frac{1}{2}\sum_{(x,y) \in S^C} F(x,y)\\
        &\leq \frac{1}{2} \sum_{(x,y) \in S} P(x,y) \left(1 - R(x,y)\right)^2 + \frac{1}{2}\sum_{(x,y) \in S^C} F(x,y)\\
        &= \frac{1}{2} \E_{P}[(1-R(X,Y))^2]  + \frac{1}{2} \sum_{(x,y) \in S^C}F(x,y).
    \end{align*}
    Note that $(X,Y) \in S^C$ is a large deviation event under $F$, so that 
    \begin{equation*}
        \sum_{(x,y) \in S^C} F(x,y) = q^{N+1} = \exp\{{-\Theta(N)}\}.
    \end{equation*}

    We have 
    \begin{align*}
        \E_P[(1-R)^2] &= \E_P\left[\left(U_X U_Y + \frac{U_X \eps_Y + U_Y \eps_X}{2} - \frac{\eps_X + \eps_Y}{2}\right)^2\right]\\
        &\overset{(a)}{\leq} 3 \left(\E_P\left[U_X^2 U_Y^2 \right] + \frac{1}{4}\E_P\left[(U_X\eps_Y + U_Y \eps_X)^2\right] + \frac{1}{4} \E_P\left[(\eps_X + \eps_Y)^2\right] \right)\\
        &\overset{(b)}{\leq} 3 \E_P[U_X^2 U_Y^2] + \frac{3}{2} \left(\E_P\left[U_X^2 \eps_Y^2\right] + \E_P\left[U_Y^2 \eps_X^2\right] + \E_P\left[\eps_X^2\right] + \E_P\left[\eps_Y^2\right]\right)\\
        &\overset{(c)}{=} 3 \left(\left(\E_P\left[U_X^2\right]\right)^2 + \E_P\left[U_X^2\right]\E_P\left[\eps_X^2\right] + \E_P\left[\eps_X^2\right]\right),
    \end{align*}
    where $(a)$ and $(b)$ use $(a+b+c)^2 \leq 3(a^2 + b^2+c^2)$ and $(a+b)^2 \leq 2(a^2+b^2)$ respectively, and $(c)$ uses the fact that $X$ and $Y$ are independent and identically distributed under $P$. 

    Additionally, 
    \begin{equation*}
        \E_P[U_X^2] =  (N (1-q))^{-2} \Var_P[X] = O\left(N^{-1}\right),
    \end{equation*}

    Now, we claim that $\E[\eps_X^2] = O\left(N^{-2}\right)$. To see this, take $E := \left\{\abs{Z_x} \leq \frac{N(1-q)}{2}\right\}$, and $\mathbbm{1}_E$ to be the indicator variable for $E$. We have $\E_P[\eps_X^2] = \E_P[\eps_X^2 \mathbbm 1_E] + \E_P[\eps^2_X \mathbbm{1}_{E^C}]$. 
    \begin{align*}
        \E_P[\eps_X^2 \mathbbm 1_E] &\leq 12 \cdot  \E_P\left[\frac{Z_X^4 + Z_X^2 + N^2 q^2 (1-q)^2 }{N^4(1-q)^4} \cdot \mathbbm 1_E\right]
    \end{align*}
    Note that 
    \begin{equation*}
        \E_P[Z_X^4] = O(N^2), \quad \E_P[Z_X^2] = O(N),
    \end{equation*}
    so
    \begin{align*}
        \E_P[\eps_X^2 \mathbbm 1_E] &\leq 12 \frac{\E_P[Z_X^4]}{N^4 (1-q)^4} + O\left(N^{-2}\right)\\
        &=O\left(N^{-2}\right)
    \end{align*}

    On the other hand, since $\abs{Z_X} \leq N$
    \begin{align*}
        \abs{\eps_X} &= \frac{\abs{Z_X^2 - Z_X - Nq(1-q)}}{N(1-q) (N + 1 - X)}\\
        &\leq \frac{N^2 + N(1+q)(1-q)}{N(1-q)} := M(N,q) = O(N).
    \end{align*}
    Then 
    \begin{equation*}
        \E_P[\eps_X^2 \mathbbm 1_{E^C}] \leq M^2(N,q) \Prob_P[E^C].
    \end{equation*}
    Using Hoeffding's inequality we have 
    \begin{align*}
        \Proba{E^C} &= \Proba{\abs{X - \E[X]} > \frac{N(1-q)}{2}} \\
        &\leq 2 \Exp{-N \frac{(1-q)^2}{2}}
    \end{align*}
    Leading to
    \begin{equation*}
        \E_P[\eps^2] \leq 2 M^2(N,q) \Exp{-N\frac{(1-q)^2}{2}} + \E_P[\eps^2 \mathbbm 1_E] = O\left(N^{-2}\right).
    \end{equation*}

    Therefore, 
    \begin{align*}
        H^2(P,F) \leq q^{N+1} + \frac{1}{2} \E_{P}[(1-R(X,Y))^2] = O\left(N^{-2}\right)
    \end{align*}
\end{proof}

\begin{proof}[Proof of  \Cref{thm:LowBound_EdgeCount} ]

    Let $K_{n_1,n_2}$ be the bipartite graph on $n_1+n_2$ vertices, which is maximally connected across the bi-partition. For $n$ even\footnote{For $n$ odd, we can add an isolated vertex and use the same process}, take $N = \frac{n}{2}$ and $G_1 = K_{N, N}$ and $G_2 = K_{N +1, N-1}$. 

    Let $g: \{0, \dots, N+1\}^2 \to \mc G^{\leq}_n$ map its inputs to the isomorphism class of complete bipartite graphs, so that $(x,y) \mapsto K_{x,y}$. Taking $(x,y) \sim P$, one can see that $g(x,y) \sim \Phi_p(G_1)$, and $(x,y) \sim F$, then $g(x,y) \sim \Phi_p(G_2)$. By the data processing inequality, we have 
    \begin{equation*}
        H^2(\Phi_p(G_1), \Phi_p(G_2)) \leq H^2(P,F),
    \end{equation*}
    so that \ref{prop:interHellDist} gives us $H^2(\Phi_p(G_1), \Phi_p(G_2)) = O(N^{-2})$.

    Let $\Phi_p(G_1)^m, \Phi_p(G_2)^m$ be the $m$-times independent product distributions, and use $H^2(\Phi_p(G_1), \Phi_p(G_2)) := h$ to denote the Hellinger distance between the distributions, so that $h \leq \frac{C}{N^2}$ for some constant $C > 0$. In order to guarantee that an estimator can distinguish between the distributions with probability $\geq 1 -\delta$ using $m$ samples, we require 
    \begin{equation}
        \frac{1}{2}\|\Phi_p(G_1)^m - \Phi_p(G_2)^m\|_1 \geq 1-2\delta.
    \end{equation}
    
    From the bound between the total variation and Hellinger distance along with tensorization \cite{Polyanskiy_Wu_2025}, we get 
    \begin{align*}
        1 - 2 \delta \leq \frac{1}{2} \oneNorm{\Phi_p(G_1)^m - \Phi_p(G_2)^m} &\leq \sqrt{1 - \left(1 - H^2(\Phi_p(G_1)^m, \Phi_p(G_2)^m)\right)^2}\\
        &= \sqrt{1 - \left(1 - h\right)^{2m}}
    \end{align*}
    Therefore, 
    \begin{align*}
        (1 - h)^{2m} &\leq 4\delta(1-\delta)\\
        \implies 2m \log\left(1-h\right) &\leq \log\left(4\delta(1-\delta)\right)
    \end{align*}
    Supposing $n$ is large enough so that $h \in [0,1/2)$, we have 
    \begin{align*}
        m &\geq \frac{\log\left( \frac{1}{4\delta (1-\delta)} \right)}{2\log\left(\frac{1}{1-h}\right)} \\
        &\geq \frac{\log\left( \frac{1}{4\delta (1-\delta)} \right)}{4h}\\
        % &= \Omega\left(N^2 \log \left(\frac{1}{\delta}\right)\right)\\
        &= \Omega\left(n^2 \log \left(\frac{1}{\delta}\right)\right).
    \end{align*}
\end{proof}

\end{document}